\documentclass[11pt,a4paper,USenglish]{article}
\usepackage{graphicx}

\usepackage[pagebackref]{hyperref}
\usepackage[utf8]{inputenc}
\usepackage{amsthm}
 \usepackage{amsmath}
 \usepackage{amssymb}
\usepackage{amsfonts}

\usepackage[numbers,sort]{natbib}
\usepackage{tabularx} %
\usepackage{algorithm}
\usepackage{algorithmicx}
\usepackage{algpseudocode}
\usepackage{pgfplots}
\usepackage{stmaryrd} %
\usepackage{paralist}
\usepackage{booktabs}
\usepackage{csvsimple}
\usepackage{pgfplotstable}
\usepackage{subfigure}
\usepackage{url}

\newcolumntype{R}[1]{>{\RaggedLeft\hspace{0pt}}p{#1}}

\newcolumntype{C}[1]{>{\centering\arraybackslash}p{#1}}

\definecolor{dark-gray}{gray}{0.2}

\usetikzlibrary{backgrounds,fit,decorations.pathreplacing,arrows,patterns,automata,calc,plotmarks}
\tikzset{
  treenode/.style = {align=center, inner sep=0pt, text centered,
    font=\sffamily},
  arn_n/.style = {treenode, circle, white, font=\sffamily\bfseries, draw=black,
    fill=black, text width=1.5em},%
  arn_r/.style = {treenode, circle, black, draw=black, dashed, 
    text width=1.5em, very thick},%
  arn_x/.style = {treenode, rectangle, draw=black,
    minimum width=0.5em, minimum height=0.5em}%
}

\newcommand{\problemdef}[3]{%
  \begin{minipage}{.9\linewidth}
    \vspace{5pt}
    \begin{compactdesc}%
    \item[\normalfont\textsc{#1}]
    \item[\emph{Input:}]#2
    \item[\emph{Task:}]#3
    \end{compactdesc}%
    \vspace{5pt}
  \end{minipage}
}

\newcommand{\sqin}{%
  \mathrel{\vphantom{\sqsubset}\text{%
    \mathsurround=0pt
    \ooalign{$\sqsubset$\cr$-$\cr}%
  }}%
}

\newcommand{\citeVLM}{\citet{viard2015computing}}

\newcommand{\dsd}{$\Delta$\nobreakdash-slice degeneracy}
\newcommand{\dsdm}{$\Delta$\nobreakdash-slice degeneracies}
\newcommand{\Dsd}{$\Delta$\nobreakdash-Slice Degeneracy}
\newcommand{\dC}{$\Delta$\nobreakdash-clique}

\newcommand{\bkd}{\textsc{BronKerboschDelta}}
\newcommand{\bkdp}{\textsc{BronKerboschDeltaPivot}}

\newcommand{\dclq}{$\Delta$\nobreakdash-clique}
\newcommand{\tmpg}{\ensuremath{\mathbb{G}}}

\newtheorem{theorem}{Theorem}
\newtheorem{lemma}{Lemma}
\newtheorem{corollary}{Corollary}

\newtheorem{proposition}{Proposition}
\newtheorem{definition}{Definition}
\newtheorem{example}{Example}

\usepackage{authblk}

\begin{document}

\title{Adapting the Bron-Kerbosch Algorithm for Enumerating Maximal Cliques in Temporal Graphs\footnote{A preliminary version of this article appeared in the Proceedings of the 2016 IEEE/ACM International Conference on Advances in Social Networks Analysis and Mining~\cite{HMNS16}. Parts of this work are based on the first author's Bachelor thesis at TU~Berlin~\cite{Him16}.}}
%


\author[1]{Anne-Sophie Himmel}

\affil[1]{Institut f\"ur Softwaretechnik und Theoretische Informatik,
 TU Berlin, Germany\\
 \texttt{anne-sophie.himmel@campus.tu-berlin.de, \{h.molter,~rolf.niedermeier\}@tu-berlin.de}}

\author[1]{Hendrik Molter}

\author[1]{Rolf Niedermeier}

\author[2]{Manuel Sorge}

\affil[2]{Ben Gurion University of the Negev, Department of Industrial Engineering and Management, Be'er Sheva, Israel\\ \texttt{sorge@post.bgu.ac.il}}

\maketitle

\begin{abstract}
Dynamics of interactions play an increasingly important role in the analysis of complex networks. A modeling framework to capture this are temporal graphs which consist of a set of vertices (entities in the network) and a set of time-stamped binary interactions between the vertices.  We focus on enumerating $\Delta$-cliques, an extension of the concept of cliques to temporal graphs: for a given time period~$\Delta$, a $\Delta$-clique in a temporal graph is a set of vertices and a time interval such that all vertices interact with each other at least after every~$\Delta$~time steps within the time interval. 
Viard, Latapy, and Magnien~[ASONAM~2015, TCS~2016] proposed a greedy algorithm for enumerating all maximal $\Delta$-cliques in temporal graphs. In contrast to this approach, we adapt the Bron-Kerbosch algorithm---an efficient, recursive backtracking algorithm which enumerates all maximal cliques in static graphs---to the temporal 
setting. We obtain encouraging results both in theory (concerning worst-case running time analysis
based on the parameter ``$\Delta$-slice degeneracy'' 
of the underlying graph) as well
as in practice\footnote{Code freely available at \url{http://fpt.akt.tu-berlin.de/temporalcliques/} (GNU General Public License).} with experiments on real-world data.
The latter culminates in an improvement for most interesting $\Delta$-values concerning running time in comparison with the algorithm of Viard, Latapy, and Magnien. 
%
%
\end{abstract}

\section{Introduction}

Network analysis is one of the main pillars of data science.
Focusing on networks that are modeled by undirected graphs,
a fundamental primitive is the identification of complete
subgraphs, that is, cliques. This is particularly true in 
the context of detecting 
communities in social networks. 
Finding a maximum-cardinality clique in a graph is a classical 
NP-hard problem, so super-polynomial worst-case running time seems unavoidable.
Moreover, often one wants to solve the more general task of not only finding 
one maximum-cardinality clique but to list \emph{all maximal} cliques.
Their number can be exponential in the graph size.
The famous Bron-Kerbosch algorithm (``Algorithm~457'' in 
\emph{Communications of the ACM~1973},~\cite{bron1973algorithm}) addresses this task and still today forms 
the basis for the best (practical) algorithms to enumerate all maximal cliques 
in  static graphs~\cite{ELS13}.
However, to realistically model many real-world phenomena in social 
and other network structures, 
one has to take into account the dynamics of the modeled system of 
interactions between entities, leading to so-called 
temporal networks. In a nutshell, compared to the standard 
static networks, the interactions in temporal networks (edges) 
appear sporadically over time (while the vertex set remains static).
Indeed, as~\citet{nicosia2013graph} pointed out, in many real-world 
systems the interactions among entities are rarely persistent over time and the non-temporal interpretation is an ``oversimplifying approximation''.
In this work, we use the standard model of temporal graphs. A temporal graph consists of a vertex set and a set of edges, each with an integer time-stamp. The generalization of a clique to the temporal setting that we study is called \dclq\ and was introduced by \citet{Viard2015Dyno,viard2015computing}. Intuitively, being in a \dclq\ means to be regularly in contact with all other entities in this \dclq. In slightly more formal terms, each pair of vertices in the \dclq\ has to be in contact in at least every $\Delta$ time steps. A fully formal definition is given in Section~\ref{sec:preliminaries}. We present an adaption of the
framework of Bron and Kerbosch to temporal graphs. 
To this end, we overcome several conceptual hurdles and propose 
a temporal version of the Bron-Kerbosch algorithm as a new standard for 
efficient enumeration of maximal \dclq s in temporal graphs.

\subsection{Related Work}
Our work relates to two main lines of research.
First, enumerating $\Delta$-cliques in temporal graphs generalizes the enumeration of maximal cliques
in static graphs, this being  subject of 
many different algorithmic approaches 
(sometimes also 
exploiting specific properties such as the ``degree of isolation'' of the cliques 
searched 
for)~\cite{bron1973algorithm,ELS13,II09,HKMN09,komusiewicz2009isolation,
tomita2006worst}. Indeed, clique finding is a special case of dense subgraph 
detection. 
Second, more recently, mining dynamic or temporal networks for 
periodic interactions~\cite{LB10} or preserving structures~\cite{uno2015mining}
(in particular, this may include cliques as a very fundamental pattern)
has gained increased attention.
Our work is directly motivated by the study of \citet{Viard2015Dyno,viard2015computing} who introduced the concept of $\Delta$-cliques
and provided a corresponding enumeration algorithm for $\Delta$-cliques.
In fact, following one of their concluding remarks on future 
research possibilities, we adapt the Bron-Kerbosch
algorithm to the temporal setting, thereby outperforming their 
greedy-based approach in most cases.

\subsection{Results and Organization}
Our main contribution is to adapt the Bron-Kerbosch 
recursive backtracking 
algorithm for clique enumeration in static graphs to temporal graphs.
In this way, we achieve a significant speedup for most interesting time period values~$\Delta$ (typically two orders 
of magnitude of speedup)
when compared to a previous algorithm due to \citet{Viard2015Dyno,viard2015computing} which is based on a greedy 
approach. We also provide a theoretical running time analysis of our 
Bron-Kerbosch adaption employing the framework of parameterized complexity 
analysis. The analysis is based on the 
parameter ``$\Delta$-slice degeneracy'' which we introduce, an adaption of the degeneracy parameter that is frequently used in static graphs as a measure for sparsity. This extends results 
concerning the static Bron-Kerbosch 
algorithm~\cite{ELS13}.
A particular feature to achieve high efficiency of the standard 
Bron-Kerbosch algorithm is the use of pivoting, a procedure to reduce the number of recursive calls of the Bron-Kerbosch algorithm. We show how to 
define this and make it work in the temporal setting, where it becomes
a significantly more delicate issue than in the static case.
In summary, we propose our temporal version of the Bron-Kerbosch approach as a 
current standard for enumerating maximal cliques in temporal graphs.

  The paper is organized as follows. In Section~\ref{sec:preliminaries}, we introduce all main definitions and notations. In addition, we give a description of the original Bron-Kerbosch algorithm as well as two extensions: pivoting and degeneracy ordering. In Section~\ref{section:bronKerboschDelta}, we propose an adaption of the Bron-Kerbosch algorithm to enumerate all maximal $\Delta$-cliques in a temporal graph, prove the correctness of the algorithm and give a running time upper bound. Furthermore, we adapt the idea of pivoting to the temporal setting. In Section~\ref{sec:degeneracy} we adapt the concept of degeneracy to the temporal setting and give an improved running time bound for enumerating all maximal $\Delta$-cliques. In Section~\ref{section:implExp}, we present the main results of the experiments on real-world data sets. We measure the $\Delta$-slice degeneracy of real-world temporal graphs, we study the efficiency of our algorithm, and compare its running time to the algorithm of \citet{Viard2015Dyno},
showing a significant performance increase due to our Bron-Kerbosch 
approach. 
We conclude in Section~\ref{sec:conclusion}, also presenting directions for future research.

\section{Preliminaries} 
\label{sec:preliminaries}
 In this section we introduce the most important notations and definitions used throughout this article. 
\subsection{Graph-Theoretic Concepts}
In the following we provide definitions of adaptations to the temporal setting for central graph-theoretic concepts.
\subsubsection{Temporal Graphs} 
  The motivation behind temporal graphs, which are also referred to as temporal networks~\cite{holme2012temporal}, time-varying graphs~\cite{nicosia2013graph}, or link streams~\cite{Viard2015Dyno}, is to capture changes in a graph that occur over time. In this work, we use the well-established model where each edge is given a time stamp~\cite{Viard2015Dyno, holme2012temporal, boccaletti2014structure}. Assuming discrete time steps, this is equivalent to a sequence of static graphs over a fixed set of vertices~\cite{michail2015introduction,Erlebach0K15%
}. Formally, the model is defined as follows.
\begin{definition}[Temporal Graph]
 A \emph{temporal graph}~$\mathbb{G}=(V,E,T)$ is defined as a triple consisting of a set of vertices~$V$, a set of \emph{time-edges}~$E \subseteq \binom{V}{2} \times T$, and a time interval~$T=[\alpha, \omega]$, where~$\alpha, \omega \in \mathbb{N}$,~$T \subseteq \mathbb{N}$ and~$\omega -\alpha$ is the \emph{lifetime} of the temporal graph~$\mathbb{G}$. 
 \end{definition} 
 The notation~$\binom{V}{2}$ describes the set of all possible undirected edges~$\{v_1,v_2\}$ with~$v_1 \not = v_2$ and~$v_1,v_2 \in V$. A time-edge~$e=(\{v_1,v_2\}, t) \in E$ can be interpreted as an interaction between~$v_1$ and~$v_2$ at time~$t$. Note that we will restrict our attention to discretized time, implying that changes only occur at discrete points in time. This seems close to a natural abstraction of real-world dynamic systems and ``gives the problems a purely combinatorial flavor''~\cite{Michail2014}.
\subsubsection{$\Delta$-Cliques}
A straightforward adaptation of a clique to the temporal setting is to additionally assign a lifetime~$I = [a, b]$ to it, that is, the largest time interval such that the clique exists in each time step in said interval. However, this model is often too restrictive for real-world data. For example, if the subject matter of examination is e-mail traffic and the data set includes e-mails with time stamps including seconds, we are not interested in people who sent e-mails to each other every second over a certain time interval, but we would like to know which groups of people were in contact with each other, say, at least every seven days over months. One possible approach would be to generalize the time stamps, taking into account only the week an e-mail was sent, resulting in a loss of accuracy in the data set. The constraint of each pair of vertices being connected in each time step can be relaxed by introducing an additional parameter~$\Delta$, quantifying how many time steps may be skipped between two connections of any vertex pair. These so-called~$\Delta$-cliques were introduced by \citet{Viard2015Dyno,viard2015computing} and are formally defined as follows.
\begin{definition}[$\Delta$-Clique]
Let~$\Delta \in \mathbb{N}$. A \emph{$\Delta$-clique} in a temporal graph~$\mathbb{G}=(V,E,T)$ is a tuple~$C=(X, I=[a,b])$ with $X \subseteq V$, $b-a\geq \Delta$, and~$I\subseteq T$ such that for all~$\tau \in [a, b-\Delta]$ and for all~$v,w\in X$ with~$v\not =w$ there exists a~$(\{v,w\},t) \in E$ with~$t \in [\tau, \tau + \Delta]$. 
\end{definition}
 In other words, for a~$\Delta$-clique~$C=(X,I)$ all pairs of vertices in~$X$ interact with each other at least after every~$\Delta$ time steps during the time interval~$I$. We implicitly exclude~$\Delta$-cliques with time intervals smaller than~$\Delta$. 
 
 It is evident that the parameter~$\Delta$ is a measurement of the intensity of interactions in~$\Delta$-cliques. Small~$\Delta$-values imply that the interaction between vertices in a~$\Delta$-clique has to be more frequent than in the case of large~$\Delta$-values. The choice of~$\Delta$ depends on the data set and the purpose of the analysis. 

 We can also consider~$\Delta$-cliques from another point of view. For a given temporal graph~$\mathbb{G}=(T,V,E)$ and a~$\Delta \in \mathbb{N}$, the static graph~$G^{\Delta}_{\tau} = (V_{\tau},E_{\tau})$ describes all contacts that appear within the~$\Delta$-sized time window~$[\tau, \tau+\Delta]$ with~$\tau \in [\alpha,\omega-\Delta]$ in the temporal graph~$\mathbb{G}$, that is~$V_{\tau} = V$ and for every~$\{v_1,v_2\} \in E_{\tau}$ there is a time step~$t \in [\tau, \tau + \Delta]$ such that~$(\{v_1,v_2\},t) \in E$. The existence of a~$\Delta$-clique~$C=(X, I=[a,b])$ indicates that all vertices in~$X$  form a clique in all static graphs~$G^{\Delta}_{\tau}$ with~$\tau \in [a,b-\Delta]$. This implies that all vertices in~$X$ are pairwise connected to each other in the static graphs of all sliding,~$\Delta$-sized time windows from time~$a$ until~$b - \Delta$. 

 By setting~$\Delta$ to the length of the whole lifetime of the temporal graph, every~$\Delta$-clique corresponds to a normal clique in the underlying static graph that results from ignoring the time stamps of the time-edges. %

We are most interested in \dC s that are not contained in any other \dC. For this we also need to adapt the notion of maximality to the temporal setting~\cite{Viard2015Dyno,viard2015computing}. Let \tmpg\ be a temporal graph. We call a $\Delta$-clique~$C=(X, I)$ in \tmpg\ \emph{vertex-maximal} if we cannot add any vertex to~$X$ without having to decrease the clique's lifetime~$I$. That is, there is no \dclq~$C' = (X', I')$ in \tmpg\ with $I \subseteq I'$ and $X \subsetneq X'$. We say that a $\Delta$-clique is \emph{time-maximal} if we cannot increase the lifetime~$I$ without having to remove vertices from~$X$. That is, there is no \dclq~$C' = (X', I')$ in \tmpg\ with $I \subsetneq I'$ and $X \subseteq X'$. We call a $\Delta$-clique \emph{maximal} if it is both vertex-maximal and time-maximal.

\subsubsection{$\Delta$-Neighborhood, $\Delta$-Cut, and other Temporal Graph Concepts}
In this section, we introduce and define further graph theoretical concepts that need to be adapted to the temporal setting.

We refer to a tuple~$(v, I=[a, b])$ with~$v \in V$ and~$I \subseteq T$ as a \emph{vertex-interval pair} of a temporal graph. We call $a$ the \emph{starting point} of interval~$I$ and $b$ the \emph{endpoint} of interval $I$. Let~$X$ be a set of vertex-interval pairs. 
The modified element relation~$(v,I) \sqin X$ (\emph{temporal membership}) expresses that there exists a vertex-interval pair~$(v,I') \in X$ with~$I \subseteq I'$.

Using these definitions, we can adapt the notion of a neighborhood of a vertex to temporal graphs. Intuitively, we want that two vertex-interval pairs are neighbors if they can be put into a \dC\ together.
\begin{definition}[$\Delta$-Neighborhood]
For a vertex~$v \in V$ and a time interval~$I \subseteq T$ in a temporal graph, the \emph{$\Delta$-neighborhood}~$N^{\Delta}(v,I)$ is the set of all vertex-interval pairs~$(w,I'=[a',b'])$ with the property that for every~$\tau \in [a',b'-\Delta]$ at least one edge~$(\{v,w\},t) \in E$ with~$t \in [\tau, \tau + \Delta]$ exists. Furthermore,~$b'-a' \geq \Delta$, $I' \subseteq I$, and~$I'$ is maximal, that is, there is no time interval~$I''\subseteq I$ with~$I' \subset I''$ satisfying the properties above.
\end{definition}
Notice that being a $\Delta$-neighbor of another vertex is a symmetric relation. If~$(w, I') \sqin N^{\Delta}(v,I)$, then we say that~$w$ is a v\emph{$\Delta$-neighbor} of~$v$ during the time interval~$I'$. In Figure~\ref{figure:TemporalGraph}, we visualize the concepts of $\Delta$-neighborhood and $\Delta$-clique in a temporal graph. See also Example~\ref{example:temporalGraph} below.

We need to define a suitable way of intersecting of two sets of vertex-interval pairs, so that, as the intuition suggests, a $\Delta$-clique is just the intersection of the ``closed'' $\Delta$-neighborhoods\footnote{In static graphs, the closed neighborhood of a vertex includes the vertex itself.} of its elements over the lifetime of the clique.
\begin{definition}[$\Delta$-Cut]
  Let~$X$ and~$Y$ be two sets of vertex-interval pairs. The \emph{$\Delta$-cut}~$X\sqcap Y$ contains for each vertex, all intersections of intervals in~$X$ and $Y$ that are of size at least~$\Delta$. More precisely, $$X \sqcap Y = \{ (v, I \cap I') \mid (v, I)\in X \wedge (v, I')\in Y \wedge |I\cap I'|\ge\Delta\}.$$
\end{definition}
In other words, the $\Delta$-cut~$X\sqcap Y$ contains all vertex-interval pairs~$(v,I)$ such that $(v, I)\sqin X$ and~$(v, I) \sqin Y$, as well as~$|I| \geq \Delta$, and~$I$ is maximal under these properties. That is, there is no~$J$ with~$I \subsetneq J$ and~$J \subseteq I'$ and~$J \subseteq I''$ such that~$(v,I') \sqin X$ and~$(v,I'') \sqin Y$ for some~$I'$ and~$I''$.

\begin{figure}[!t]
\begin{center}
\subfigure[$N^{\Delta}(a,T)$]{
        \begin{tikzpicture}[thick]
    \tikzstyle{phase} = [fill,shape=circle,minimum size=5pt,inner sep=0pt]
    \node (time1) at (0,0.5) {0};
    \node (time2) at (1,0.5) {1};
    \node (time3) at (2,0.5) {2};
    \node (time4) at (3,0.5) {3};
    \node (time5) at (4,0.5) {4};
    \node (time6) at (5,0.5) {5};
    \node (time1) at (6,0.5) {6};
    \node (time1) at (7,0.5) {7};
    \node (time1) at (8,0.5) {8};
    \node at (-1,0) (q1) {a};
    \node at (-1,-1) (q2) {b};
    \node at (-1,-2) (q3) {c};
    \node (node0a) at (0,0) {};
    \node (node0b) at (0,-1) {};
    \node (node0c) at (0,-2) {};
    \node[phase] (node2a) at (2,0) {};
    \node[phase] (node2b) at (2,-1) {};
    \node (node2c) at (3,-1) {};
    \node (node2c) at (2,-2) {};
    \draw[-] (node2a) to[out=-60,in=60] (node2b);
    \node[phase] (node3a) at (3,0) {};
    \node[phase] (node3b) at (3,-1) {};
    \draw[-] (node3a) to[out=-60,in=60] (node3b);
    \node[phase] (node4a) at (4,0) {};
    \node[phase] (node4c) at (4,-2) {};
    \draw[-] (node4a) to[out=-60,in=60] (node4c);
    \node (node5a) at (5,0) {};
    \node[phase] (node5b) at (5,-1) {};
    \node[phase] (node5c) at (5,-2) {};
    \draw[-] (node5b) to[out=-60,in=60] (node5c);
    \node[phase] (node6a) at (6,0) {};
    \node[phase] (node6c) at (6,-2) {};
    \draw[-] (node6a) to[out=-60,in=60] (node6c);
    \node (node7a) at (7,0) {};
    \node (node7b) at (7,-1) {};
    \node (node7c) at (7,-2) {};
    \node (node8a) at (8,0) {};
    \node (node8b) at (8,-1) {};
    \node (node8c) at (8,-2) {};
    \node (end1) at (9,0) {} edge [-] (q1);
    \node (end2) at (9,-1) {} edge [-] (q2);
    \node (end3) at (9,-2) {} edge [-] (q3);
    \begin{pgfonlayer} {background}
    \node (background) [fill=green, fill opacity=0.5, fit = (node2c) (node8c), postaction={pattern=north east lines, pattern color=gray}] {};
    \node (background) [fill=green, fill opacity=0.5, fit = (node4a) (node4c), postaction={pattern=north east lines, pattern color=gray}] {};
    \node (background) [fill=green, fill opacity=0.5, fit = (node6a) (node6c), postaction={pattern=north east lines, pattern color=gray}] {};
    \node (background) [fill=yellow , fill opacity=0.5, fit = (node0b) (node5b)] {};
    \node (background) [fill=yellow , fill opacity=0.5, fit = (node2a) (node2b)] {};
    \node (background) [fill=yellow , fill opacity=0.5, fit = (node3a) (node3b)] {};
    \end{pgfonlayer}
    \end{tikzpicture}
    \label{figure:DeltaNeighborhoodA}
}
\subfigure[$N^{\Delta}(b,T)$]{
 \begin{tikzpicture}[thick]
    \tikzstyle{phase} = [fill,shape=circle,minimum size=5pt,inner sep=0pt]
    \node (time1) at (0,0.5) {0};
    \node (time2) at (1,0.5) {1};
    \node (time3) at (2,0.5) {2};
    \node (time4) at (3,0.5) {3};
    \node (time5) at (4,0.5) {4};
    \node (time6) at (5,0.5) {5};
    \node (time1) at (6,0.5) {6};
    \node (time1) at (7,0.5) {7};
    \node (time1) at (8,0.5) {8};
    \node at (-1,0) (q1) {a};
    \node at (-1,-1) (q2) {b};
    \node at (-1,-2) (q3) {c};
    \node (node0a) at (0,0) {};
    \node (node0b) at (0,-1) {};
    \node (node0c) at (0,-2) {};
    \node[phase] (node2a) at (2,0) {};
    \node[phase] (node2b) at (2,-1) {};
    \node (node2c) at (3,-1) {};
    \node (node2c) at (2,-2) {};
    \draw[-] (node2a) to[out=-60,in=60] (node2b);
    \node[phase] (node3a) at (3,0) {};
    \node[phase] (node3b) at (3,-1) {};
    \node (node3c) at (3,-2) {};
    \draw[-] (node3a) to[out=-60,in=60] (node3b);
    \node[phase] (node4a) at (4,0) {};
    \node[phase] (node4c) at (4,-2) {};
    \draw[-] (node4a) to[out=-60,in=60] (node4c);
    \node (node5a) at (5,0) {};
    \node[phase] (node5b) at (5,-1) {};
    \node[phase] (node5c) at (5,-2) {};
    \draw[-] (node5b) to[out=-60,in=60] (node5c);
    \node[phase] (node6a) at (6,0) {};
    \node[phase] (node6c) at (6,-2) {};
    \draw[-] (node6a) to[out=-60,in=60] (node6c);
    \node (node7a) at (7,0) {};
    \node (node7b) at (7,-1) {};
    \node (node7c) at (7,-2) {};
    \node (node8a) at (8,0) {};
    \node (node8b) at (8,-1) {};
    \node (node8c) at (8,-2) {};
    \node (end1) at (9,0) {} edge [-] (q1);
    \node (end2) at (9,-1) {} edge [-] (q2);
    \node (end3) at (9,-2) {} edge [-] (q3);
    \begin{pgfonlayer} {background}
    \node (background) [fill=green, fill opacity=0.5, fit = (node3c) (node7c), postaction={pattern=north east lines, pattern color=gray}] {};
    \node (background) [fill=green, fill opacity=0.5, fit = (node5b) (node5c), postaction={pattern=north east lines, pattern color=gray}] {};
    \node (background) [fill=yellow , fill opacity=0.5, fit = (node0a) (node5a)] {};
    \node (background) [fill=yellow , fill opacity=0.5, fit = (node2a) (node2b)] {};
    \node (background) [fill=yellow , fill opacity=0.5, fit = (node3a) (node3b)] {};
    \end{pgfonlayer}
    \end{tikzpicture}
    \label{figure:DeltaNeighborhoodB}
}\\
\subfigure[$N^{\Delta}(c,T)$]{
 \begin{tikzpicture}[thick]
    \tikzstyle{phase} = [fill,shape=circle,minimum size=5pt,inner sep=0pt]
    \node (time1) at (0,0.5) {0};
    \node (time2) at (1,0.5) {1};
    \node (time3) at (2,0.5) {2};
    \node (time4) at (3,0.5) {3};
    \node (time5) at (4,0.5) {4};
    \node (time6) at (5,0.5) {5};
    \node (time1) at (6,0.5) {6};
    \node (time1) at (7,0.5) {7};
    \node (time1) at (8,0.5) {8};
    \node at (-1,0) (q1) {a};
    \node at (-1,-1) (q2) {b};
    \node at (-1,-2) (q3) {c};
    \node (node0a) at (0,0) {};
    \node (node0b) at (0,-1) {};
    \node (node0c) at (0,-2) {};
    \node[phase] (node2a) at (2,0) {};
    \node[phase] (node2b) at (2,-1) {};
    \node (node2c) at (3,-1) {};
    \node (node2c) at (2,-2) {};
    \draw[-] (node2a) to[out=-60,in=60] (node2b);
    \node[phase] (node3a) at (3,0) {};
    \node[phase] (node3b) at (3,-1) {};
    \node (node3c) at (3,-2) {};
    \draw[-] (node3a) to[out=-60,in=60] (node3b);
    \node[phase] (node4a) at (4,0) {};
    \node[phase] (node4c) at (4,-2) {};
    \draw[-] (node4a) to[out=-60,in=60] (node4c);
    \node (node5a) at (5,0) {};
    \node[phase] (node5b) at (5,-1) {};
    \node[phase] (node5c) at (5,-2) {};
    \draw[-] (node5b) to[out=-60,in=60] (node5c);
    \node[phase] (node6a) at (6,0) {};
    \node[phase] (node6c) at (6,-2) {};
    \draw[-] (node6a) to[out=-60,in=60] (node6c);
    \node (node7a) at (7,0) {};
    \node (node7b) at (7,-1) {};
    \node (node7c) at (7,-2) {};
    \node (node8a) at (8,0) {};
    \node (node8b) at (8,-1) {};
    \node (node8c) at (8,-2) {};
    \node (end1) at (9,0) {} edge [-] (q1);
    \node (end2) at (9,-1) {} edge [-] (q2);
    \node (end3) at (9,-2) {} edge [-] (q3);
    \begin{pgfonlayer} {background}
    \node (background) [fill=green, fill opacity=0.5, fit = (node3b) (node7b), postaction={pattern=north east lines, pattern color=gray}] {};
    \node (background) [fill=green, fill opacity=0.5, fit = (node5b) (node5c), postaction={pattern=north east lines, pattern color=gray}] {};
    \node (background) [fill=yellow , fill opacity=0.5, fit = (node2a) (node8a)] {};
    \node (background) [fill=yellow , fill opacity=0.5, fit = (node4a) (node4c)] {};
    \node (background) [fill=yellow , fill opacity=0.5, fit = (node6a) (node6c)] {};
    \end{pgfonlayer}
    \end{tikzpicture}
    \label{figure:DeltaNeighborhoodC}
}
\subfigure[Maximal~$\Delta$-Clique~$(\{a,b,c\} \mathbin{,} \lbrack 3 \mathbin{,} 5 \rbrack)$]{
    \begin{tikzpicture}[thick]
    \tikzstyle{phase} = [fill,shape=circle,minimum size=5pt,inner sep=0pt]
    \node (time1) at (0,0.5) {0};
    \node (time2) at (1,0.5) {1};
    \node (time3) at (2,0.5) {2};
    \node (time4) at (3,0.5) {3};
    \node (time5) at (4,0.5) {4};
    \node (time6) at (5,0.5) {5};
    \node (time1) at (6,0.5) {6};
    \node (time1) at (7,0.5) {7};
    \node (time1) at (8,0.5) {8};
    \node at (-1,0) (q1) {a};
    \node at (-1,-1) (q2) {b};
    \node at (-1,-2) (q3) {c};
    \node (node0a) at (0,0) {};
    \node (node0b) at (0,-1) {};
    \node (node0c) at (0,-2) {};
    \node[phase] (node2a) at (2,0) {};
    \node[phase] (node2b) at (2,-1) {};
    \node (node2c) at (3,-1) {};
    \node (node2c) at (2,-2) {};
    \draw[-] (node2a) to[out=-60,in=60] (node2b);
    \node[phase] (node3a) at (3,0) {};
    \node[phase] (node3b) at (3,-1) {};
    \node (node3c) at (3,-2) {};
    \draw[-] (node3a) to[out=-60,in=60] (node3b);
    \node[phase] (node4a) at (4,0) {};
    \node[phase] (node4c) at (4,-2) {};
    \draw[-] (node4a) to[out=-60,in=60] (node4c);
    \node (node5a) at (5,0) {};
    \node[phase] (node5b) at (5,-1) {};
    \node[phase] (node5c) at (5,-2) {};
    \draw[-] (node5b) to[out=-60,in=60] (node5c);
    \node[phase] (node6a) at (6,0) {};
    \node[phase] (node6c) at (6,-2) {};
    \draw[-] (node6a) to[out=-60,in=60] (node6c);
    \node (node7a) at (7,0) {};
    \node (node7b) at (7,-1) {};
    \node (node7c) at (7,-2) {};
    \node (node8a) at (8,0) {};
    \node (node8b) at (8,-1) {};
    \node (node8c) at (8,-2) {};
    \node (end1) at (9,0) {} edge [-] (q1);
    \node (end2) at (9,-1) {} edge [-] (q2);
    \node (end3) at (9,-2) {} edge [-] (q3);
    \begin{pgfonlayer} {background}
    \node (background) [fill=yellow , fill opacity=0.5, fit = (node3a) (node5c)] {};
    \end{pgfonlayer}
    \end{tikzpicture}
   \label{figure:DeltaClique}
}
  \caption{$\Delta$-Neighborhoods and a~$\Delta$-clique of a temporal graph with~$\Delta = 2$. The lifetime of the graph is~$T=[0,8]$. The elements of the $\Delta$-neighborhoods in (a), (b), and (c) are shaded in yellow and green (hatched), respectively. A maximal $\Delta$-clique (d) is shaded in yellow.}
    \label{figure:TemporalGraph}
\end{center}
\end{figure}
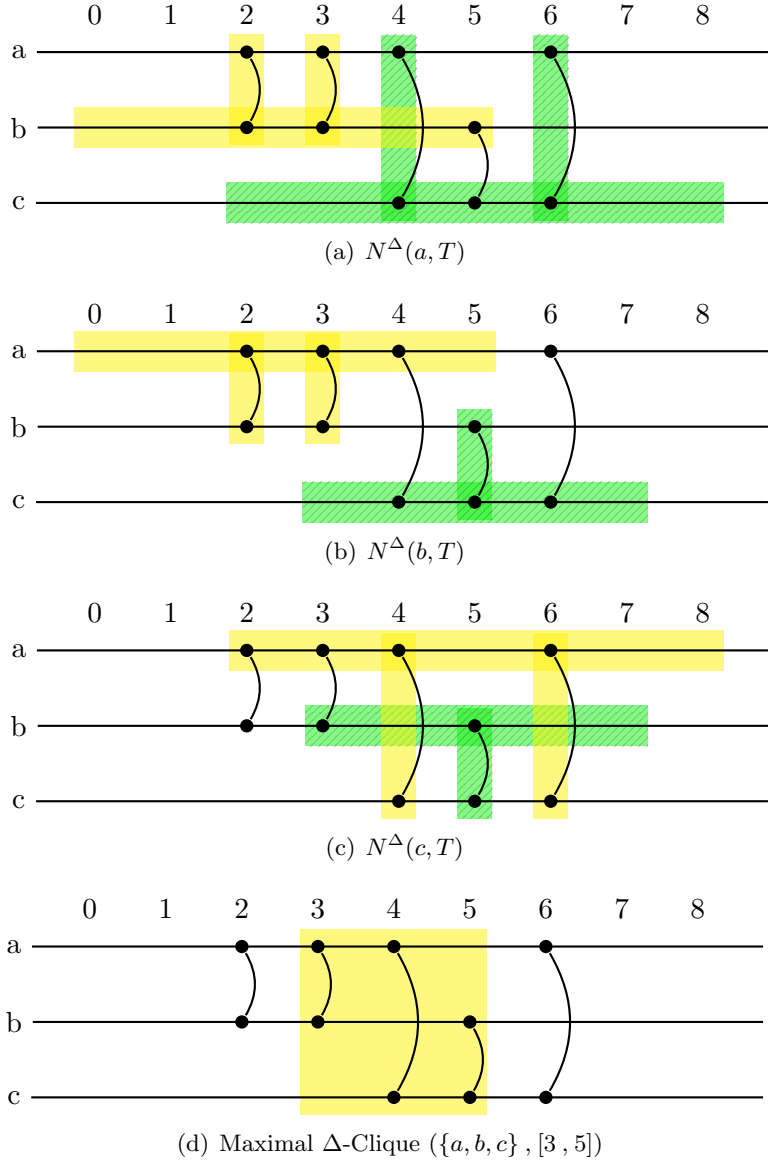
\begin{example}
\label{example:temporalGraph}
In Figure \ref{figure:TemporalGraph} we visualize a temporal graph and the concepts of $\Delta$-neighborhood and $\Delta$-clique. We consider a temporal graph $\mathbb{G} = (T, V, E)$ with $T = [0,8]$, $V = \{a,b,c\}$, $E=\{(\{a,b\},2), (\{a,b\},3), (\{a,c\},4), (\{b,c\},5),$ $(\{a,c\},6)\}$, and $\Delta = 2$. The vertices are visualized as horizontal lines. The connections between two vertices at a specific time step represent the time-edges of the temporal graph.

We visualize the $\Delta$-neighborhood of each vertex of the temporal graph over the whole time interval $T$ in Figures~\ref{figure:DeltaNeighborhoodA}-\ref{figure:DeltaNeighborhoodC}:
\begin{itemize}
\item In Figure~\ref{figure:DeltaNeighborhoodA}, we consider the $\Delta$-neighborhood $N^{\Delta}(a,T)$ of vertex $a$ during the whole time interval $T$. The yellow shaded bar marks the vertex-interval pair $(b,[0,5]) \in N^{\Delta}(a,T)$. The vertex $b$ is a $\Delta$-neighbor of $a$ during $[0,5]$ because for every $\tau \in [0,5-\Delta =3]$ at least one time-edge~$(\{a,b\},t) \in E$ with $t \in [\tau, \tau + \Delta]$ exists since $(\{a,b\},2), (\{a,b\},3) \in E$. The same holds for the vertex-interval pair $(c,[2,8]) \in  N^{\Delta}(a,T)$ which is marked in hatched green.

\item In Figure~\ref{figure:DeltaNeighborhoodB}, we visualize the $\Delta$-neighborhood $N^{\Delta}(b,T)$ of $b$ over the whole lifetime $T$ of the temporal graph. The vertex-interval pair $(c,[3,7]) \in N^{\Delta}(b,T)$ is marked in hatched green. The vertex-interval pair $(a,[0,5]) \in N^{\Delta}(b,T)$ is shaded in yellow. It becomes evident that being a $\Delta$-neighbor of another vertex is a symmetric relation---if $a$ is a $\Delta$-neighbor of $b$ during $[0,5]$, then $b$ is also a $\Delta$-neighbor of $a$ during~$[0,5]$.

\item In Figure~\ref{figure:DeltaNeighborhoodC}, we visualize the $\Delta$-neighborhood $N^{\Delta}(c,T)$ of $c$ over the whole lifetime $T$ of the temporal graph. The vertex-interval pair $(b,[3,7]) \in N^{\Delta}(c,T)$ is marked in hatched green. The vertex-interval pair $(a,[2,8]) \in N^{\Delta}(c,T)$ is shaded in yellow.
\end{itemize}
Figure \ref{figure:DeltaClique} shows the maximal $\Delta$-clique $(\{a,b,c\},[3,5])$. During the time interval $[3,5]$, $a$ and $b$ are $\Delta$-neighbors, $b$ and $c$ are $\Delta$-neighbors and $a$ and~$c$ are $\Delta$-neighbors, see Figures~\ref{figure:DeltaNeighborhoodA}-\ref{figure:DeltaNeighborhoodC}. We cannot increase the time interval because at time step $2$ the vertices $b$ and $c$ are not yet $\Delta$-neighbors and at time step $6$ the vertices $a$ and $b$ are no longer $\Delta$-neighbors. Further nontrivial maximal \dC s in this temporal graph are: $(\{a,b\},[0,5])$, $(\{a,c\},[2,8])$, $(\{b,c\},[3,7])$, as well as the trivial $\Delta$-cliques $(\{a\},[0,8])$, $(\{b\},[0,8])$, and $(\{c\},[0,8])$.
\end{example}

\subsection{Bron-Kerbosch Algorithm}
\label{section:bronKerbosch}

In this section, we explain the basic idea of the (static) Bron-Kerbosch algorithm. We also present two techniques known from the literature which improve the running time of the algorithm.

The Bron-Kerbosch algorithm~\cite{bron1973algorithm} enumerates all maximal cliques in undirected, static graphs. It is a widely used recursive backtracking algorithm which is easy to implement and more efficient than alternative algorithms in many practical applications~\cite{ELS13}.

\begin{algorithm}[t]
\begin{algorithmic}[1] %
\Function{BronKerbosch}{$P, R, X$}
 \Comment{ $R:$ a clique}
 \Comment{ $P \cup X:$ set of all vertices $v$ such that $R \cup \{v\}$ is a clique and where \begin{itemize}
   \item vertices in $P$ have not yet been considered as additions to $R$ and 
   \item vertices in $X$ already have been considered in earlier steps
 \end{itemize}}
\If{$P \cup X = \emptyset$} 
      \State{add~$R$ to the solution}
\EndIf
\For{$v \in P$}
	\State{\Call{BronKerbosch}{$P \cap N(v), R \cup \{v\}, X \cap N(v))$}}
	\State $P \gets P \setminus \{v\}$
	\State $X \gets X \cup \{v\}$
\EndFor
\EndFunction
\end{algorithmic}
\caption{Enumerating all Maximal Cliques}
\label{alg:bronker}
\end{algorithm}
The Bron-Kerbosch algorithm, see Algorithm~\ref{alg:bronker}, receives three disjoint vertex-sets as an input:~$P$,~$R$, and~$X$. The set~$R$ induces a clique and~$P \cup X$ is the set of all vertices which are adjacent to every vertex in~$R$. Each vertex in~$P \cup X$ is a witness that the clique~$R$ is not maximal yet. The set~$P$ contains the vertices that have not been considered yet whereas the set~$X$ includes all vertices that have already been considered in earlier steps. %
In each recursive call, the algorithm checks whether the given clique~$R$ is maximal or not. If~$P \cup X = \emptyset$, then there are no vertices that can be added to the clique and therefore, the clique is maximal and can be added to the solution. 
 Otherwise, the clique is not maximal because at least one vertex exists that is adjacent to all vertices in $R$ and consequently would form a clique with $R$. For each $v \in P$ the algorithm makes a recursive call for the clique $R \cup \{v\}$ and restricts~$P$ and~$X$ to the neighborhood of~$v$. After the recursive call, vertex $v$ is removed from $P$ and added to $X$. This guarantees that the same maximal cliques are not detected multiple times.
For a graph~$G=(V,E)$ the algorithm is initially called with~$P=V$ and~$R=X=\emptyset$.

\subsubsection{Pivoting}
\label{subsection:BKPivoting}

\citet{bron1973algorithm} introduced a method to increase the efficiency of the basic algorithm by choosing a pivot element to decrease the number of recursive calls. It is based on the observation that for any vertex~$u \in P \cup X$ either~$u$ itself or one of its non-neighbors must be contained in any maximal clique containing~$R$. 
 This is true since if neither~$u$ nor one of the non-neighbors of~$u$ are included in a clique containing~$R$, then this clique cannot be maximal because~$u$ can be added to this clique due to the fact that only neighbors of~$u$ were added to~$R$.
Hence, if we modify the Bron-Kerbosch algorithm (Algorithm~\ref{alg:bronker}) so that we choose an arbitrary pivot element~$u \in P \cup X$ and iterate only over~$u$ and all its non-neighbors, then we still enumerate all maximal cliques containing~$R$ but decrease the number of recursive calls in the for-loop of Algorithm~\ref{alg:bronker}. %
\citet{tomita2006worst} have shown that if~$u$ is chosen from~$P \cup X$ such that~$u$ has the most neighbors in~$P$, then all maximal cliques of a graph~$G=(V,E)$ are enumerated in $O(3^{\mid V\mid /3})$~time, see Algorithm~\ref{alg:bronkerpivot}.
\begin{algorithm}[t]
\begin{algorithmic}[1] %
\Function{BronKerboschPivot}{$P, R, X$}
 \Comment{ $R:$ a clique} 
 \Comment{ $P \cup X:$ set of all vertices $v$ such that $R \cup \{v\}$ is a clique and where \begin{itemize}
   \item vertices in $P$ have not yet been considered as additions to $R$ and 
   \item vertices in $X$ already have been considered in earlier steps
 \end{itemize}}
\If{$P \cup X = \emptyset$}
	\State{add $R$ to the solution}
\EndIf
\State{choose pivot vertex $u \in P \cup X$ with $|P \cap N(u)| = \smash{\displaystyle \max_{v \in P \cup X}} \mid P \cap N(v) \mid$}
\For{$v \in P \setminus N(u)$}
	\State{\Call{BronKerboschPivot}{$P \cap N(v), R \cup \{v\}, X \cap N(v))$}}
	\State $P \gets P \setminus \{v\}$
	\State $X \gets X \cup \{v\}$
\EndFor
\EndFunction
\end{algorithmic}
\caption{Enumerating all Maximal Cliques in a Graph with Pivoting}
\label{alg:bronkerpivot}
\end{algorithm}

\subsubsection{Degeneracy of a Graph} 
\label{subsection:BKDegeneracy}
Degeneracy is a measure of graph sparsity. Real-world instances of static graphs (especially social networks) are often sparse, resulting in a small degeneracy value~\cite{ELS13}. This motivates a modification of the Bron-Kerbosch algorithm which we present in this section and the complexity analysis of this algorithm parameterized by the degeneracy of the input graph. The degeneracy of a graph is defined as follows.

\begin{definition}[Degeneracy]
The \emph{degeneracy} of a static graph~$G$ is defined as the smallest integer $d \in \mathbb{N}$ such
that each subgraph~$G'$ of~$G$ contains a vertex~$v$ with degree at most $d$.
\end{definition}
If a graph has degeneracy $d$, we also call it \emph{$d$-degenerated}. It is easy to see that the maximal clique
size of a $d$-degenerated graph is at most~$d+1$: If there is a clique of size at
least~$d+2$, then the vertices of this clique would form a subgraph in
which every vertex~$v$ of the clique has a degree larger than $d$.
For each~$d$-degenerated graph there is a \emph{degeneracy ordering}, which is a linear ordering of the vertices with the property that for every vertex~$v$ we have that at most~$d$ of its neighbors occur at a later position in the ordering. The degeneracy~$d$ and a corresponding degeneracy ordering for a graph~$G=(V,E)$  can be computed in linear time~\cite{ELS13}: For graph~$G$, the vertex with the smallest degree is selected in each step and removed from the graph until no vertex is left. The degeneracy of the graph is the highest degree of a vertex at the time the vertex has been removed from the graph and a corresponding degeneracy ordering is the order in which the vertices were removed from the graph.

For a graph~$G=(V,E)$ with degeneracy~$d$, \citet{ELS13} showed that using the degeneracy ordering of~$G$ in the outer-most recursive call and afterwards using pivoting, all maximal cliques can be enumerated in~$O(d \cdot | V | \cdot3^{d/3})$ time, see Algorithm~\ref{alg:bronkerdeg}. In other words, enumerating maximal cliques is fixed-parameter tractable with respect to the parameter degeneracy~$d$ of the input graph.
\begin{algorithm}[t]
\begin{algorithmic}[1] %
\Function{BronKerboschDeg}{$P, R, X$}
 \Comment{ $R:$ a clique}
 \Comment{ $P \cup X:$ set of all vertices $v$ such that $R \cup \{v\}$ is a clique and where \begin{itemize}
   \item vertices in $P$ have not yet been considered as additions to $R$ and 
   \item vertices in $X$ already have been considered in earlier steps
 \end{itemize}}
\For{$v_i$ in a degeneracy odering $v_0, v_1, \dots , v_n$ of $G=(V,E)$}
	\State $P \gets N(v_i) \cap \{v_{i+1}, \ldots ,v_{n-1}\}$
	\State $X \gets N(v_i) \cap \{v_{0}, \ldots ,v_{i-1}\}$
	\State{\Call{BronKerboschPivot}{$P, \{v_i\}, X $}}
\EndFor
\EndFunction
\end{algorithmic}
\caption{Enumerating all Maximal Cliques in a Graph with Degeneracy Ordering}
\label{alg:bronkerdeg}
\end{algorithm}

\section{Bron-Kerbosch Algorithm for Temporal Graphs}
\label{section:bronKerboschDelta}

We adapt the static Bron-Kerbosch algorithm to the temporal setting to enumerate all $\Delta$-cliques, see Algorithm~\ref{alg:bronkerdelta}. 
The input of the algorithm consists of two sets~$P$~and~$X$ of vertex-interval pairs as well as a tuple~$R=(C, I)$, where~$C$ is a set of vertices and~$I$ a time interval. The idea is that in every recursive call of the algorithm,~$R$ is a time-maximal $\Delta$-clique, and the sets~$P$~and~$X$ contain vertex-interval pairs that are in the $\Delta$-neighborhood of every vertex in~$C$ during an interval~$I' \subseteq I$.
In particular, $P \cup X$ includes all vertex-interval pairs~$(v,I)$ for which~$(C \cup \{v\}, I)$ is a time-maximal $\Delta$-clique.
While each vertex-interval pair in~$P$ still has to be combined with~$R$ to ensure that every maximal $\Delta$-clique will be found, for every vertex-interval pair~$(v,I') \in X$ every maximal $\Delta$-clique~$(C', I'')$ with~$C \cup \{v\} \subseteq C'$ and~$I'' \subseteq I'$ has already been detected in earlier steps.

\begin{algorithm}[t]
\begin{algorithmic}[1] %
\Function{\bkd}{$P ,R=(C,I) , X$} 
 \Comment{ $R=(C,I):$ time-maximal $\Delta$-clique}
 \Comment{ $P \cup X:$ set of all vertex-interval pairs $(v,I')$ such that $I' \subseteq I$ and $(C \cup \{v\}, I')$ is a time-maximal $\Delta$-clique and where
 \begin{itemize}
   \item vertex-interval pairs in $P$ have not yet been considered as additions to $R$ and
   \item vertex-interval pairs in $X$ already have been considered in earlier steps
 \end{itemize}}
\If{$\forall (w, I') \in P \cup X \colon I' \subsetneq I$}
	\State{add~$R$ to the solution}
\EndIf
\For{$(v,I') \in P$}
	\State{$R' \gets (C \cup \{v\}, I')$}
	\State{$P' \gets P \sqcap N^{\Delta}(v,I')$}
	\State{$X' \gets X  \sqcap N^{\Delta}(v,I')$}
	\State{\Call{\bkd}{$P',R', X'$}}
	\State{$P \gets P \setminus \{(v,I')\}$}
	\State{$X \gets X \cup \{(v,I')\}$}
\EndFor
\EndFunction
\end{algorithmic}
\caption{Enumerating all Maximal $\Delta$-Cliques}
\label{alg:bronkerdelta}
\end{algorithm}

We show below that if~$\forall (w, I') \in P \cup X \colon I' \subsetneq I$, then there is no vertex~$v$ that forms a $\Delta$-clique together with~$C$ over the whole time interval~$I$. Consequently,~$R=(C, I)$ is a maximal $\Delta$-clique. 

In one step, for every vertex-interval pair $(v,I') \in P$ a recursive call is initiated for the $\Delta$-clique $R'=(C \cup \{v\}, I')$ with all parameters restricted to the $\Delta$-neighborhood of $v$ in the time interval $I'$, that is, $P \sqcap N^{\Delta}(v,I')$ and $X \sqcap N^{\Delta}(v,I')$. For the set $P'$ for example, we get a set of all time-maximal vertex-interval pairs $(w,I'')$ for which it holds that $(w,I'') \sqin N^{\Delta}(v,I')$ and $(w,I'') \sqin P$. This restriction is made so that for all $(w,I'') \in P'$ of the recursive call the vertex $w$ is not only a $\Delta$-neighbor of all $x \in C$ but also of the vertex~$v$ during the time $I'' \subseteq I'$.

After the recursive call for $\Delta$-clique $(C \cup \{v\}, I')$, the tuple $(v,I')$ is removed from the set $P$ and added to the set $X$ to avoid that the same cliques are found multiple times.

For a temporal graph~$\mathbb{G}=(V,E,T)$ and a given time period~$\Delta$, the \emph{initial call} of Algorithm \ref{alg:bronkerdelta} to enumerate all maximal $\Delta$-cliques in graph~$\mathbb{G}$ is made with $P = \{(v,T) \mid v \in V\}$, $R = (\emptyset, T)$ and~$X=\emptyset$. In the remainder of this document we will always assume that \bkd{} is initially called with those inputs.

\subsection{Analysis}
In the following, we prove the correctness of the algorithm and analyze its running time. We start with arguing that the sets~$P$ and~$X$ behave as claimed.
\begin{lemma}
\label{lemma:setPX}
For each recursive call of \bkd{} with~$R=(C,I)$ and~$C\neq\emptyset$, it holds that~$P \cup X = \bigsqcap_{v \in C} N^{\Delta}(v,I)$.
\end{lemma}
\begin{proof}
We prove this by induction on the recursion depth, that is, the number~$|C|$ of vertices in the clique in the current recursive call. In the initial call we have that~$C = \emptyset$. In each iteration of the first call we have that~$P \cup X = \{(v,T) \mid v \in V\}$ since, whenever a vertex-interval pair is removed from~$P$, then it is added to~$X$, and initially~$P = \{(v,T) \mid v \in V\}$. For every recursive call of \bkd{} with~$R'=(C', I')$,~$P'$, and~$X'$, and~$C' = \{v\}$ for some vertex~$v$ we have that $P'=P\sqcap N^{\Delta}(v,I')$ and~$X' = X  \sqcap N^{\Delta}(v,I')$. Hence, we get
\[
P' \cup X' = \{(v,T) \mid v \in V\} \sqcap N^{\Delta}(v,I') = N^{\Delta}(v,I').
\]
Now we assume that we are in a recursive call of \bkd{} with~$R=(C, I)$,~$P$, and~$X$, where~$|C| > 1$.
By the induction hypothesis we know that $P \cup X = \bigsqcap_{v \in C} N^{\Delta}(v,I)$.
Let~$(v,I') \in P$ be the vertex added to the $\Delta$-clique, that is, in the next recursive call we have that~$R'=(C', I')$, with~$C' = C \cup \{v\}$, and~$P' = P \sqcap N^{\Delta}(v,I')$ as well as~$X' = X  \sqcap N^{\Delta}(v,I')$. Then,
\begin{align*}
	P' \cup X'
	&= (P \sqcap N^{\Delta}(v,I'))\cup (X\sqcap N^{\Delta}(v,I'))\\
 	&= (P\cup X)\sqcap N^{\Delta}(v,I')\\
 	&= \underset{w \in C}{\bigsqcap} N^{\Delta}(w,I) \sqcap N^{\Delta}(v,I')\\
	&= \underset{w \in C'}{\bigsqcap} N^{\Delta}(w,I').
 \end{align*}
This proves the claim.
\end{proof}
Next, we show that the set~$R$ behaves as claimed, that is,~$R$ is indeed a time-maximal $\Delta$-clique in each recursive call of \bkd{}.
\begin{lemma}
\label{lemma:timeMax}
In each recursive call of \bkd{},~$R=(C, I)$ is a time-maximal $\Delta$-clique.
\end{lemma}
\begin{proof}
We show by induction on the recursion depth that~$R=(C, I)$ is a time-maximal $\Delta$-clique and that all vertex-interval pairs~$(v, I')$ in~$P$ are $\Delta$\nobreakdash-neighbors during~$I'$ to all vertices in the $\Delta$-clique~$R$ and that~$I'$ is maximal under this property. The algorithm is initially called with~$R = (\emptyset, T)$, which is a trivial time-maximal $\Delta$-clique, and~$P = \{(v,T) \mid v \in V\}$, which fulfills the desired property since the initial $\Delta$-clique is empty and~$T$ is the maximum time interval. In each recursive call \bkd{} is called with~$(P \sqcap N^{\Delta}(v,I'), (C \cup \{v\}, I'), X  \sqcap N^{\Delta}(v,I'))$ for some~$(v, I')\in P$. By the induction hypothesis,~$v$ is a $\Delta$-neighbor to all vertices in~$C$ during time interval~$I'$, and~$I'$~is maximal. Hence, $(C \cup \{v\}, I')$ is a time-maximal \dC. Furthermore, each vertex-interval pair~$(v', I'')$ in~$P \sqcap N^{\Delta}(v,I')$ is in the $\Delta$\nobreakdash-neighbor\-hood of each vertex-interval pair~$(v'', I')$ with~$v'' \in C \cup \{v\}$, since it is both in~$P$ and hence in the $\Delta$-neighborhood of each vertex in~$C$ and in~$N^{\Delta}(v,I')$. The maximality of~$I'$ follows from the fact that the $\Delta$-cut and $\Delta$-neighborhood operations preserve maximality of intervals by definition.
\end{proof}
Now we can prove the correctness of the algorithm.
\begin{theorem}[Correctness of Algorithm~\ref{alg:bronkerdelta}]
\label{thm:correctness}
Let~$\mathbb{G}=(V,E,T)$ be a temporal graph. If algorithm \bkd{}$(P, R, X)$ is run on input~$(V \times \{T\}, (\emptyset, T), \emptyset)$, then it adds all maximal $\Delta$-cliques of~$\mathbb{G}$, and only these, to the solution. 
\end{theorem}
\begin{proof}
  Let~$R^*=(C^*, I^*)$ be a maximal $\Delta$-clique
  with~$|C^*|>1$. For a recursive call of \bkd\ on $(P, R, X)$, say
  that a vertex is a \emph{candidate}, if there is an interval $I$
  with $I^* \subseteq I$ such that $(v, I) \in P$. We show by
  induction on $|C^*| - \ell$, that for each
  $\ell = 0, 1, \ldots, |C^*|$ there is a recursive call of \bkd\ on
  $(P, R = (C, I), X)$ with $C \subseteq C^*$ and
  $\ell = |C^* \setminus C|$ candidates.

  Clearly, in the initial call, $C = \emptyset \subseteq C^*$ and each
  vertex in $C^*$ is a candidate. Now assume that there is a recursive
  call with~$(P, R = (C, I), X)$ and~$C \subseteq C^*$, and with
  $\ell - 1 = |C^* \setminus C|$ candidates. Consider the for-loop in
  that recursive call and consider the first vertex-interval
  pair~$(v, I')$ in that loop in which $v$ is a candidate and
  $I^* \subseteq I'$. \bkd\ proceeds with a recursive call on
  $(P \sqcap N^{\Delta}(v, I'), R' = (C \cup \{v\}, I'), X
  \sqcap N^{\Delta}(v, I'))$. Observe that each candidate except
  $v$ remains a candidate also in this recursive call. Furthermore,~$|C^* \setminus (C \cup \{v\})| = \ell$. Thus, by induction there is
  a recursive call with the sets~$(P, R, X)$ in which $R^* =
  R$. Since $R^*$ is maximal by assumption, for each vertex-interval pair $(w, I'') \in \bigsqcap_{v \in C^*} N^{\Delta}(v, I^*)$ we have $I'' \subsetneq I^*$. By
  Lemma~\ref{lemma:setPX} we have $\forall (w, I') \in P \cup X \colon I' \subsetneq I^*$ and hence, $R^*$
  is added to the solution.
  
  Now assume that $R=(C,I)$ is added to the solution. By Lemma~\ref{lemma:timeMax}, $R$ is a time-maximal $\Delta$-clique. By Lemma~\ref{lemma:setPX} and since $P \cup X = \emptyset$, there is no vertex that can be added to~$R$. Hence, $R$ is a maximal $\Delta$-clique.
\end{proof}
Next, we analyze the running time of \bkd{}. We start with the following observation.
\begin{lemma}
\label{lemma:cliquecount}
For every time-maximal $\Delta$-clique~$R$ of a temporal graph~$\mathbb{G}=(V,E,T)$, there is at most one recursive call of \bkd{} with~$R$ as an input.
\end{lemma}
\begin{proof}
  Assume that there are two recursive calls $A$ and $B$ of \bkd{} with the same~$R=(C, I)$. Let~$R'=(C', I')$, with~$C'\subset C$ and~$I \subseteq I'$, occur in the recursive call corresponding to the closest common ancestor of the recursive calls $A$ and $B$ in the recursion tree. Hence, there are two vertex-interval pairs~$(v, J), (w, J') \in P$ that lead to the calls~$A$ and~$B$, respectively, in the for loop.

  Consider the case $v = w$. Then, $J$ and $J'$ must overlap in at least $\Delta$ time steps, because $I \subseteq J, J' \subseteq I'$. However, $P$ is contained in the $\Delta$-cut of the $\Delta$-neighborhoods of $C'$ over $I'$ and thus, for each vertex no two vertex-interval pairs in $P$ overlap in $\Delta$ time steps, a contradiction.

  Now consider the case $v \neq w$. Without loss of generality due to symmetry assume that $(v, J)$ is processed first in the for loop. Then, when processing~$(w, J')$, pair~$(v, J)$ has been added to~$X$. This is a contradiction to the fact that recursive call~$B$ outputs $R$, that is, it outputs a clique with time interval~$I \subseteq J$.

  Hence, we have that there cannot be two recursive calls of \bkd{} with~$R=(C, I)$.
\end{proof}
Now we upper-bound the running time for computing a $\Delta$-cut.
\begin{lemma}\label{lemma:deltacut}
 
Let~$X$ and~$Y$ be two sets of vertex-interval pairs with the following properties.
\begin{compactitem}
\item For every~$(v, I)\in X\cup Y$ we have that $|I|\ge \Delta$,
  \item for every~$(v, I)\in X$ and~$(v, I')\in X$ we have that~$|I\cap I'|<\Delta$,
  \item for every~$(v, I)\in Y$ and~$(v, I')\in Y$ we have that~$|I\cap I'|<\Delta$,
  \item $X$ and~$Y$ are sorted lexicographically by first the vertex and then the starting point of the interval.
\end{compactitem}
Then the~$\Delta$-cut~$X \sqcap Y$ can be computed in~$O(|X|+|Y|)$ time such that it is also sorted lexicographically by first the vertex and then the starting point of the interval.
\end{lemma}
\begin{proof}
The $\Delta$-cut~$X \sqcap Y$ of two sets of vertex-interval pairs~$X$ and~$Y$ can be computed in the following way. 

 For every vertex $v$, we do the following:
\begin{compactenum}
\item Select the first vertex-interval pairs~$(v, I)$ and~$(v, I')$ from~$X$ and~$Y$, respectively.  
\item If~$|I\cap I'|>\Delta$, then add~$(v, I\cap I')$ to the output (the $\Delta$-cut). If the endpoint of~$I'$ is smaller than the endpoint of~$I$, then replace $(v, I')$ with the next vertex-interval pair in~$Y$, otherwise replace $(v, I)$ with the next vertex-interval pair in~$X$. 
\item Repeat Step~2 until all vertex-interval pairs containing vertex~$v$ are processed.
\end{compactenum}
Note that the intervals for each vertex~$v$ are added to the output in order of their starting point. Furthermore, by construction of the algorithm we have that for each $(v, I)$ in the output, $(v, I)$ is also in the $\Delta$-cut~$X \sqcap Y$. It remains to show that for all $(v, I)\in X$ and $(v, I')\in Y$ with $|I\cap I'|\ge\Delta$ we have that $(v, I\cap I')$ is included in the output. Let $I=[a, b]$ and $I' = [a', b']$. At some point, the procedure processes in Step 2 for the first time one of $(v, I) \in X$ or $(v, I') \in Y$. Without loss of generality, let $(v, I) \in X$ be processed first. If at the same time also $(v, I') \in Y$ is processed, clearly, $(v, I \cap I')$ is added to the output, as required. Now assume that Step 2 processes some other vertex-interval pair $(v, I''=[a'', b''])\in Y$, $a'' < a'$, together with $(v, I) \in X$. Since $|I\cap I'|\ge\Delta$ and $|I'\cap I''|<\Delta$ we have that $b''<b$ and hence, $(v, I)$ is not replaced in this step. Consequently, the procedure eventually adds~$(v, I\cap I')$ to the output.

In each step of the procedure at least one new vertex-interval pair is processed and each vertex-interval pair in $X$ and $Y$ is only processed once. Hence, the running time is in~$O(|X|+|Y|)$.
\end{proof}
Lemmata~\ref{lemma:cliquecount} and~\ref{lemma:deltacut} allow us to upper-bound the running time of \bkd{} depending on the number of different time-maximal $\Delta$-cliques of the input graph.
\begin{theorem}
\label{lemma:runningtime}
Let~$\mathbb{G}=(V,E,T)$ be a temporal graph with~$x$ distinct time-maximal $\Delta$-cliques. Then \bkd{} enumerates all \emph{maximal} $\Delta$-cliques in~$O(x\cdot |E| + |E|\cdot |T|)$ time.
\end{theorem}
\begin{proof}
We assume that all edges of the temporal graph are sorted by their time stamp. Note that this can be done in a preprocessing step in~$O(|E|\cdot |T|)$ time using Counting Sort. Furthermore, we assume that for each vertex~$v$, the $\Delta$-neighborhood~$N^{\Delta}(v,T)$ is given. These neighborhoods can be precomputed in~$O(|E|)$ time, assuming that the edges are sorted by their time stamps.

By Lemma~\ref{lemma:cliquecount} we know that for each time-maximal $\Delta$-clique there is at most one recursive call of \bkd{}.
By charging the computation of~$P'$, $R'$, and $X'$ to the corresponding recursive call, for each recursive call we compute a constant number of $\Delta$-neighborhoods and $\Delta$-cuts. The size of the sets~$P$,~$X$, and any $\Delta$-neighborhood is upper-bounded by~$|E|$ and each of these sets has the property that for every~$(v, I)$ and~$(v, I')$ out of the same set we have that~$|I\cap I'|<\Delta$. Given~$N^{\Delta}(v,T)$,~$N^{\Delta}(v,I)$ can be computed in~$O(|E|)$~time for any~$I$ and by Lemma~\ref{lemma:deltacut}, a $\Delta$-cut can be computed in~$O(|E|)$ time. Hence, all maximal $\Delta$-cliques can be enumerated in~$O(x\cdot|E| + |E|\cdot |T|)$ time.
\end{proof}
We now use a general upper bound for the number of time-maximal $\Delta$-cliques in a temporal graph to bound the overall running time of \bkd{}. 
\begin{corollary}
  Let $\mathbb{G}=(V,E,T)$ be a temporal graph. \bkd{} enumerates all maximal $\Delta$-cliques of $\mathbb{G}$ in~$O(2^{|V|} \cdot |T| \cdot |E|)$ time.
\end{corollary}
\begin{proof}
Note that the vertex set of each maximal $\Delta$-clique induces a static clique in the static graph~$G$ underlying~$\mathbb{G}$ that has an edge between two vertices if and only if there is a time-edge in~$\mathbb{G}$ between these vertices at some time step. Furthermore, for each clique in $G$, there are at most $|T|$ maximal $\Delta$-cliques because their time intervals are pairwise not contained in one-another. Hence, the number of time-maximal $\Delta$-cliques of any temporal graph is upper-bounded by~$2^{|V|} \cdot |T|$. By Theorem~\ref{lemma:runningtime}, we get an overall running time in~$O(2^{|V|} \cdot |T| \cdot |E|)$.
\end{proof}
\subsection{Pivoting}
\label{subsection:pivotingBKD}
In this section, we explain how we can decrease the number of recursive calls of \bkd\ by using pivoting. Recall that the idea of pivoting in the classic Bron-Kerbosch algorithm for static graphs is based on the observation that for any vertex~$u \in P \cup X$ either~$u$ itself or one of its non-neighbors must be contained in any maximal clique containing~$R$. Vertex~$u$ is also called \emph{pivot}. 

A similar observation holds for maximal $\Delta$-cliques in temporal graphs. Instead of vertices, however, we now choose vertex-interval pairs as pivots: For any~$(v_p,I_p) \in P \cup X$ and any maximal $\Delta$-clique~$R_{\max} = (C_{\max}, I_{\max})$ with~$I_{\max} \subseteq I_p$, either vertex~$v_p$ or one vertex~$w\neq v_p$ which is not a $\Delta$-neighbor of~$v_p$ during the time~$I_{\max}$, that is,~$(w,I_{\max}) \not \sqin N^{\Delta}(v_p,I_p)$, must be contained in~$C_{\max}$. 
By choosing a pivot element~$(v_p,I_p) \in X \cup P$ we only have to iterate over all elements in~$P$ which are not in the $\Delta$-neighborhood of the pivot element, see Algorithm~\ref{alg:bronkerdeltaPivot}. In other words, we do not have to make a recursive call for any~$(w,I') \in P$ which holds~$(w,I') \sqin N^{\Delta}(v_p,I_p)$. 

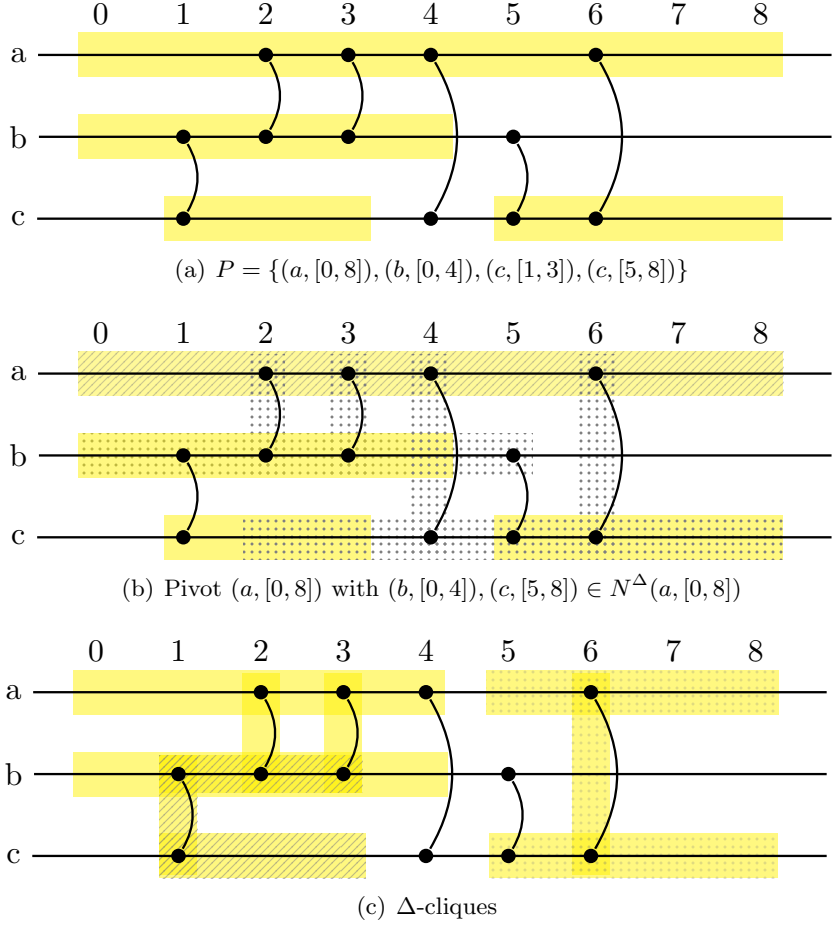
\begin{figure}[t]%
\begin{center}
\subfigure[$P = \{(a,\lbrack 0,8 \rbrack), (b,\lbrack 0,4\rbrack),(c,\lbrack 1,3\rbrack),(c,\lbrack 5,8\rbrack)\}$]{
 \resizebox{0.9\textwidth}{!}{%
        \begin{tikzpicture}[thick]
    \tikzstyle{phase} = [fill,shape=circle,minimum size=5pt,inner sep=0pt]
    \node (time0) at (0,0.5) {0};
    \node (time1) at (1,0.5) {1};
    \node (time2) at (2,0.5) {2};
    \node (time3) at (3,0.5) {3};
    \node (time4) at (4,0.5) {4};
    \node (time5) at (5,0.5) {5};
    \node (time6) at (6,0.5) {6};
    \node (time7) at (7,0.5) {7};
    \node (time8) at (8,0.5) {8};
    \node at (-1,0) (q1) {a};
    \node at (-1,-1) (q2) {b};
    \node at (-1,-2) (q3) {c};
    \node (node0a) at (0,0)  {};
    \node (node0b) at (0,-1) {};
    \node (node0c) at (0,-2) {};
    \node (node1c) at (1,-2) {};
    \node (node6a) at (6,0)  {};
    \node (node4b) at (4,-1) {};
    \node (node3c) at (3,-2) {};
    \node (node5c) at (5,-2) {};
    \node (node7c) at (7,-2) {};
    \node (node8a) at (8,0)  {};
    \node (node8b) at (8,-1) {};
    \node (node8c) at (8,-2) {};

    \node[phase] (node1b) at (1,-1) {};
    \node[phase] (node1c) at (1,-2) {};
    \draw[-] (node1b) to[out=-60,in=60] (node1c);    
    \node[phase] (node2a) at (2,0) {};
    \node[phase] (node2b) at (2,-1) {};
    \draw[-] (node2a) to[out=-60,in=60] (node2b);
    \node[phase] (node3a) at (3,0) {};
    \node[phase] (node3b) at (3,-1) {};
    \draw[-] (node3a) to[out=-60,in=60] (node3b);
    \node[phase] (node4a) at (4,0) {};
    \node[phase] (node4c) at (4,-2) {};
    \draw[-] (node4a) to[out=-60,in=60] (node4c);
    \node[phase] (node5b) at (5,-1) {};
    \node[phase] (node5c) at (5,-2) {};
    \draw[-] (node5b) to[out=-60,in=60] (node5c);
    \node[phase] (node6a) at (6,0) {};
    \node[phase] (node6c) at (6,-2) {};
    \draw[-] (node6a) to[out=-60,in=60] (node6c);
    \node (end1) at (9,0) {} edge [-] (q1);
    \node (end2) at (9,-1) {} edge [-] (q2);
    \node (end3) at (9,-2) {} edge [-] (q3);    
    \begin{pgfonlayer} {background}
    \node (background) [fill=yellow, fill opacity=0.5, fit = (node0a) (node8a)] {};
    \node (background) [fill=yellow , fill opacity=0.5, fit = (node0b) (node4b)] {};
    \node (background) [fill=yellow, fill opacity=0.5, fit = (node5c) (node8c)] {};
    \node (background) [fill=yellow, fill opacity=0.5, fit = (node1c) (node3c)] {};
    \end{pgfonlayer}
    \end{tikzpicture}
    \label{figure:PElements}
  }
}
\subfigure[Pivot $(a, \lbrack 0,8 \rbrack)$ with $(b,\lbrack 0,4\rbrack),(c,\lbrack 5,8\rbrack) \in  N^{\Delta}(a, \lbrack 0,8 \rbrack)$]{
 \resizebox{0.9\textwidth}{!}{%
        \begin{tikzpicture}[thick]
    \tikzstyle{phase} = [fill,shape=circle,minimum size=5pt,inner sep=0pt]
    \node (time0) at (0,0.5) {0};
    \node (time1) at (1,0.5) {1};
    \node (time2) at (2,0.5) {2};
    \node (time3) at (3,0.5) {3};
    \node (time4) at (4,0.5) {4};
    \node (time5) at (5,0.5) {5};
    \node (time6) at (6,0.5) {6};
    \node (time7) at (7,0.5) {7};
    \node (time8) at (8,0.5) {8};
    \node at (-1,0) (q1) {a};
    \node at (-1,-1) (q2) {b};
    \node at (-1,-2) (q3) {c};
    \node (node0a) at (0,0)  {};
    \node (node0b) at (0,-1) {};
    \node (node0c) at (0,-2) {};
    \node (node1c) at (1,-2) {};
    \node (node2c) at (2,-2) {};
    \node (node6a) at (6,0)  {};
    \node (node4b) at (4,-1) {};
    \node (node3c) at (3,-2) {};
    \node (node5c) at (5,-2) {};
    \node (node7c) at (7,-2) {};
    \node (node8a) at (8,0)  {};
    \node (node8b) at (8,-1) {};
    \node (node8c) at (8,-2) {};

    \node[phase] (node1b) at (1,-1) {};
    \node[phase] (node1c) at (1,-2) {};
    \draw[-] (node1b) to[out=-60,in=60] (node1c);    
    \node[phase] (node2a) at (2,0) {};
    \node[phase] (node2b) at (2,-1) {};
    \draw[-] (node2a) to[out=-60,in=60] (node2b);
    \node[phase] (node3a) at (3,0) {};
    \node[phase] (node3b) at (3,-1) {};
    \draw[-] (node3a) to[out=-60,in=60] (node3b);
    \node[phase] (node4a) at (4,0) {};
    \node[phase] (node4c) at (4,-2) {};
    \draw[-] (node4a) to[out=-60,in=60] (node4c);
    \node[phase] (node5b) at (5,-1) {};
    \node[phase] (node5c) at (5,-2) {};
    \draw[-] (node5b) to[out=-60,in=60] (node5c);
    \node[phase] (node6a) at (6,0) {};
    \node[phase] (node6c) at (6,-2) {};
    \draw[-] (node6a) to[out=-60,in=60] (node6c);
    \node (end1) at (9,0) {} edge [-] (q1);
    \node (end2) at (9,-1) {} edge [-] (q2);
    \node (end3) at (9,-2) {} edge [-] (q3);    
    \begin{pgfonlayer} {background}
    \node (background) [fill=yellow , fill opacity=0.4, fit = (node0a) (node8a),postaction={pattern=north east lines, pattern color=gray}] {};
    \node (background) [ fit = (node0b) (node5b), postaction={pattern=dots, pattern color=gray}] {};
    \node (background) [ fit = (node2a) (node2b), , postaction={pattern=dots, pattern color=gray}] {};
    \node (background) [ fit = (node3a) (node3b), , postaction={pattern=dots, pattern color=gray}] {};
     \node (background) [fill=yellow , fill opacity=0.5, fit = (node0b) (node4b), ] {};
    \node (background) [fill=yellow, fill opacity=0.5, fit = (node5c) (node8c)] {};
    \node (background) [fill=yellow, fill opacity=0.5, fit = (node1c) (node3c)] {};
    \node (background) [fit = (node2c) (node8c), postaction={pattern=dots, pattern color=gray}] {};
    \node (background) [ fit = (node4a) (node4c), postaction={pattern=dots, pattern color=gray}] {};
    \node (background) [ fit = (node6a) (node6c), postaction={pattern=dots, pattern color=gray}] {};
    \end{pgfonlayer}
    \end{tikzpicture}
    \label{figure:PivotElementWithDeltaNeighborhood}
  }
}
\subfigure[$\Delta$-cliques]{
 \resizebox{0.9\textwidth}{!}{%
        \begin{tikzpicture}[thick]
    \tikzstyle{phase} = [fill,shape=circle,minimum size=5pt,inner sep=0pt]
    \node (time0) at (0,0.5) {0};
    \node (time1) at (1,0.5) {1};
    \node (time2) at (2,0.5) {2};
    \node (time3) at (3,0.5) {3};
    \node (time4) at (4,0.5) {4};
    \node (time5) at (5,0.5) {5};
    \node (time6) at (6,0.5) {6};
    \node (time7) at (7,0.5) {7};
    \node (time8) at (8,0.5) {8};
    \node at (-1,0) (q1) {a};
    \node at (-1,-1) (q2) {b};
    \node at (-1,-2) (q3) {c};
    \node (node0a) at (0,0)  {};
    \node (node0b) at (0,-1) {};
    \node (node0c) at (0,-2) {};
    \node (node1c) at (1,-2) {};
    \node (node2c) at (2,-2) {};
    \node (node5a) at (5,0)  {};
    \node (node6a) at (6,0)  {};
    \node (node4b) at (4,-1) {};
    \node (node3c) at (3,-2) {};
    \node (node5c) at (5,-2) {};
    \node (node7c) at (7,-2) {};
    \node (node8a) at (8,0)  {};
    \node (node8b) at (8,-1) {};
    \node (node8c) at (8,-2) {};

    \node[phase] (node1b) at (1,-1) {};
    \node[phase] (node1c) at (1,-2) {};
    \draw[-] (node1b) to[out=-60,in=60] (node1c);    
    \node[phase] (node2a) at (2,0) {};
    \node[phase] (node2b) at (2,-1) {};
    \draw[-] (node2a) to[out=-60,in=60] (node2b);
    \node[phase] (node3a) at (3,0) {};
    \node[phase] (node3b) at (3,-1) {};
    \draw[-] (node3a) to[out=-60,in=60] (node3b);
    \node[phase] (node4a) at (4,0) {};
    \node[phase] (node4c) at (4,-2) {};
    \draw[-] (node4a) to[out=-60,in=60] (node4c);
    \node[phase] (node5b) at (5,-1) {};
    \node[phase] (node5c) at (5,-2) {};
    \draw[-] (node5b) to[out=-60,in=60] (node5c);
    \node[phase] (node6a) at (6,0) {};
    \node[phase] (node6c) at (6,-2) {};
    \draw[-] (node6a) to[out=-60,in=60] (node6c);
    \node (end1) at (9,0) {} edge [-] (q1);
    \node (end2) at (9,-1) {} edge [-] (q2);
    \node (end3) at (9,-2) {} edge [-] (q3);    
    \begin{pgfonlayer} {background}
    \node (background) [fill=yellow, fill opacity=0.5, fit = (node0a) (node4a)] {};
    \node (background) [fill=yellow, fill opacity=0.5, fit = (node0b) (node4b)] {};
    \node (background) [fill=yellow, fill opacity=0.5, fit = (node2a) (node2b)] {};
    \node (background) [fill=yellow, fill opacity=0.5, fit = (node3a) (node3b)] {};

    \node (background) [fill=yellow , fill opacity=0.5, fit = (node5a) (node8a),postaction={pattern=dots, pattern color=gray}] {};
    \node (background) [fill=yellow , fill opacity=0.5, fit = (node5c) (node8c),postaction={pattern=dots, pattern color=gray}] {};
    \node (background) [fill=yellow , fill opacity=0.5, fit = (node6a) (node6c),postaction={pattern=dots, pattern color=gray}] {};
    
    \node (background) [fill=yellow, fill opacity=0.5, fit = (node1b) (node3b), postaction={pattern=north east lines, pattern color=gray}] {};
    \node (background) [fill=yellow, fill opacity=0.5, fit = (node1c) (node3c), postaction={pattern=north east lines, pattern color=gray}] {};
    \node (background) [fill=yellow, fill opacity=0.5, fit = (node1b) (node1c), postaction={pattern=north east lines, pattern color=gray}] {};

    \end{pgfonlayer}
    \end{tikzpicture}
    \label{figure:DeltaCliques}
  }
}
\caption{A exemplary set $P$ of \bkdp{}, a possible pivot element (hatched) including its $\Delta$-neighborhood (dotted), and all maximal $\Delta$-cliques with respect to set~$P$,~$\Delta~=~2$.}
    \label{figure:Pivoting}
\end{center}
\end{figure}
In Figure~\ref{figure:Pivoting} we give an illustrative example for pivoting.
In this example, we assume that the algorithm runs on a temporal graph such that the set~$P= \{(a,[0,8 ]),(b,[0,4]),(c,[1,3]),(c,[5,8])\}$ occurs within a recursive call of \bkdp{}. For simplicity, we show in Figure~\ref{figure:PElements} only the subgraph containing the elements of $P$ and the relation between these elements rather than displaying the whole graph. 
In Figure~\ref{figure:PivotElementWithDeltaNeighborhood}, we choose element~$(a,[0,8])$ (hatched) as pivot. It can be seen that the elements~$(b,[0,4])$ and $(c,[5,8])$ lie completely in the $\Delta$-neighborhood (dotted) of the pivot, that is, $(b,[0,4]),(c,[5,8]) \sqin  N^{\Delta}(a, [0,8])$. These two elements can therefore be left out in the iteration over the elements in $P$ of the \bkd. We only have to iterate over the pivot~$(a,[0,8])$ and the element~$(c,[1,3])$ which is not completely in the $\Delta$-neighborhood of our chosen pivot. 
In Figure~\ref{figure:DeltaCliques}, we can see that for every maximal \dC~$(C,I)$ with respect to~$P$ either~$a \in C$, $I \subseteq [0,8]$ or~$c \in C$, $I \subseteq [1,3]$. The figure hence shows that iterating over the elements $(b,[0,4])$ and $(c,[5,8])$ will not find any maximal \dC\ that we do not find via one of the elements $(a,[0,8])$ and $(c,[1,3])$.

Next, we formally prove the correctness of this procedure.%
\begin{proposition}
\label{lemma:pivoting}
For each $\Delta$-clique $R=(C,I)$ and a pivot element $(v_p,I_p) \in P \cup X$, the following holds: for every  $R_{\max} = (C_{\max}, I_{\max})$ with $C \subset C_{\max}$ and~$I_{\max} \subseteq I_p \subseteq I$ it either holds that $v_p \in C_{\max}$ or otherwise there is a vertex~$w \in C_{\max}$ that satisfies $(w,I') \in P \cup X $, $I_{\max} \subseteq I'$, and $(w,I_{\max}) \not \sqin N^{\Delta}(v_p,I_p)$, and consequently $(w,I') \not \sqin N^{\Delta}(v_p,I_p)$.
\end{proposition}
\begin{proof}
Let $R_{\max} = (C_{\max}, I_{\max})$ be a maximal $\Delta$-clique with $C \subset C_{\max}$ and~$I_{\max} \subseteq I_p \subseteq I$. Assume that $v_p \notin C_{\max}$ and for each $w \in C_{\max}$ it holds that~$(w ,I_{\max}) \sqin N^{\Delta}(v_p,I_p)$. Consequently, for each $w \in C_{\max} \setminus C$ there exists a $(w,I') \in P \cup X$ with $I_{\max} \subseteq I'$. Because $v_p$ is a $\Delta$-neighbor of all vertices in~$C_{\max} \setminus C$ at least during~$I_{\max}$ and a $\Delta$-neighbor of all vertices in $C$ during~$I_p$, the vertex~$v_p$ can be added to the $\Delta$-clique $R_{\max}$, yielding another \dC\ with at least the same lifetime. This is a contradiction to the assumption that~$R_{\max}$ is maximal.
\end{proof}

\begin{algorithm}[t]
\begin{algorithmic}[1]%
\Function{BronKerboschDeltaPivot}{$P, R=(C,I), X$}
 \Comment{ $R=(C,I):$ time-maximal $\Delta$-clique}
 \Comment{ $P \cup X:$ set of all  $(v,I')$ s.t.\ $I' \subseteq I$ and $(C \cup \{v\}, I')$ is a time-maximal $\Delta$-clique}
\If{$\forall (w, I') \in P \cup X \colon I' \subsetneq I$}
	\State{add~$R$ to the solution}
\EndIf
\State choose pivot element $(v_p,I_p) \in P \cup X$
\For{$(v,I') \in P \setminus \{(w,I'') \mid (w,I'') \in P \wedge (w, I'') \sqin N^{\Delta}(v_p,I_p) \}$}
	\State{$R' \gets (C \cup \{v\}, I')$}
	\State{$P' \gets P \sqcap N^{\Delta}(v,I')$}
	\State{$X' \gets X  \sqcap N^{\Delta}(v,I')$}
	\State{\Call{BronKerboschDeltaPivot}{$P', R', X'$}}

	\State{$P \gets P \setminus (v,I')$}
	\State{$X \gets X \cup (v,I')$}
\EndFor
\EndFunction
\end{algorithmic}
\caption{Enumerating all Maximal $\Delta$-Cliques in a Temporal Graph with Pivoting}
\label{alg:bronkerdeltaPivot}
\end{algorithm}

An optimal pivot element is chosen in such a way that it minimizes the number of recursive calls. It is the element in the set~$P \cup X$ having the largest number of elements in~$P$ in its $\Delta$-neighborhood. We have seen that the whole procedure is quite similar to pivoting in the basic Bron-Kerbosch algorithm but with one difference: we are able to choose more than one pivot element. The only condition that has to be satisfied is that the time intervals of the pivot elements cannot overlap:

  For each $\Delta$-clique~$R=(C,I)$ in a recursive call of the algorithm, choosing a pivot element~$(v_p,I_p) \in P \cup X$ only affects maximal $\Delta$-cliques~$R_{\max} = (C_{\max}, I_{\max})$ fulfilling~$I_{\max} \subseteq I_p$. Moreover, for all elements~$(w,I') \in P$ satisfying~$(w,I') \sqin N^{\Delta}(v_p,I_p)$ it holds~$I' \subseteq I_p$. Consequently, a further pivot element~$(v_p',I_p') \in P\cup X$ fulfilling that~$I_p'$ does not overlap with~$I_p$ neither interferes with the considered maximal $\Delta$-cliques nor with the vertex-interval pairs in~$P$ that are in the $\Delta$-neighborhood of the pivot element~$(v_p,I_p)$. 

The problem of finding the optimal set of pivot elements in~$P \cup X$ can be formulated as a weighted interval scheduling maximization problem:

\problemdef{Weighted Interval Scheduling}{A set~$J$ of jobs~$j$ with a time interval~$I_j$ and a weight~$w_j$ each.}{Find a subset of jobs~$J' \subseteq J$ that maximizes~$\sum_{j \in J'}w_j$ such that for all~$i, j \in J'$ with~$i \not = j$, the time intervals~$I_i$ and~$I_j$ do not overlap.}

In our problem, the jobs are the elements of~$P \cup X$ and the weight of an element is thereby the number of all elements that are in~$P$ and lie in the $\Delta$-neighborhood of this element. Formally, the jobs are the elements~$(v,I') \in P \cup X$, the corresponding time interval is~$I'$ of the element~$(v,I')$ and the corresponding weight~$w_{(v,I')}= | \{(v,I) \mid (v,I) \in P \wedge (v,I) \sqin N^{\Delta}(v,I')\} |$. This problem can be solved efficiently in~$O(\min(|E|, |V|\cdot |T|) \cdot \log (\min(|E|, |V|\cdot |T|)))$ time by using dynamic programming~\cite[Chapter~6.1]{kleinberg2006algorithm} under the assumption that the weights of the potential pivot elements are known.
\section{Degeneracy of Temporal Graphs}
\label{sec:degeneracy}
  Recall from Section~\ref{subsection:BKDegeneracy} that one can upper-bound the running time of the static Bron-Kerbosch algorithm using the degeneracy of the input graph. The degeneracy of a graph~$G$ is the smallest integer~$d$ such that every non-empty subgraph of~$G$ contains a vertex of degree at most~$d$. We now give an analogue for the temporal setting, motivated by the fact that static graphs are often sparse in practice as measured by small degeneracy~\cite{ELS13}. Intuitively, we want to capture the fact that a temporal graph keeps its degeneracy value during its whole lifetime.

\begin{definition}[$\Delta$-slice degeneracy]
A temporal graph~$\mathbb{G}=(V,E,T)$ has $\Delta$-slice degeneracy~$d$ if for all~$t \in T$ we have that the graph~$G_t=(V, E_t)$, where~$E_t = \{\{v, w\} \ | \ (\{u, w\}, t') \in E \text{ for some } t' \in [t, t+\Delta]\}$, has degeneracy at most~$d$.
\end{definition}

Using the parameter $\Delta$-slice degeneracy, we can upper-bound the number of time-maximal $\Delta$-cliques of a temporal graph.%

\begin{lemma}
\label{lemma:cliquecountdeg}
Let~$\mathbb{G}=(V,E,T)$ be a temporal graph with $\Delta$-slice degeneracy~$d$. Then, the number of time-maximal $\Delta$-cliques in~$\mathbb{G}$ is at most~$3^{d/3}\cdot 2^{d+1}\cdot |V|\cdot |T|$.
\end{lemma}
\begin{proof}
Let~$\mathbb{G}=(V,E,T)$ be a temporal graph with $\Delta$-slice degeneracy~$d$. Then we call the graph~$G_t=(V, E_t)$, where~$E_t = \{\{v, w\} \ | \ (\{u, w\}, t') \in E \text{ for some } t' \in [t, t+\Delta]\}$ a \emph{$\Delta$-slice} of~$\mathbb{G}$ at time~$t$. The vertex set of each time-maximal $\Delta$-clique which starts at time~$t$ is also a clique in~$G_t$, otherwise there would be two vertices which are disconnected for more than~$\Delta$ time-steps. Since~$G_t$ has degeneracy at most~$d$, the number of maximal cliques of~$G_t$ is upper-bounded by~$3^{d/3}\cdot |V|$~\cite{ELS13}. Furthermore, the maximum size of a clique is upper-bounded by~$d+1$. Hence, the total number of cliques is upper-bounded by~$3^{d/3}\cdot 2^{d+1}\cdot |V|$. Note that for each of those cliques we have at most one time-maximal $\Delta$-clique starting at time~$t$. Hence, the total number of $\Delta$-cliques is at most~$3^{d/3}\cdot 2^{d+1}\cdot |V|\cdot|T|$.
\end{proof}
Lemma~\ref{lemma:cliquecountdeg} now allows us to bound the running time of Algorithm~\ref{alg:bronkerdelta} using the $\Delta$-slice degeneracy~$d$ of the input graph~$\mathbb{G}$.
\begin{theorem}
\label{thm:fpt}
Let~$\mathbb{G}=(V,E,T)$ be a temporal graph with $\Delta$-slice degeneracy~$d$. Then, \bkd{} enumerates all $\Delta$-cliques of~$\mathbb{G}$ in $O(3^{d/3}\cdot 2^d\cdot |V|\cdot|T|\cdot |E|)$ time.
\end{theorem}
\begin{proof}
By Lemma~\ref{lemma:cliquecountdeg} we know that the number of time-maximal $\Delta$-cliques in a temporal graph with $\Delta$-slice degeneracy~$d$ is at most~$3^{d/3}\cdot 2^{d+1}\cdot |V|\cdot |T|$. Hence, by Theorem~\ref{lemma:runningtime}, we get an overall running time in~$O(3^{d/3}\cdot 2^d\cdot |V|\cdot|T|\cdot |E|)$.
\end{proof}
\noindent Note that Theorem~\ref{thm:fpt} implies that enumerating all maximal $\Delta$-cliques is fixed-parameter tractable with respect to the parameter $\Delta$-slice degeneracy. Hence, while NP-hard in general, the problem can be solved efficiently if the $\Delta$-slice degeneracy of the input graph is small. 
\section{Experimental Results}
\label{section:implExp}
\newcommand{\ds}[1]{\textsf{#1}}
\newcommand{\hsa}{\ds{highschool-2011}} %
\newcommand{\hsb}{\ds{highschool-2012}} %
\newcommand{\hsc}{\ds{highschool-2013}} %
\newcommand{\aus}{\ds{as733}} %
\newcommand{\ema}{\ds{karlsruhe}} %
\newcommand{\soc}{\ds{facebook-like}} %
\newcommand{\enr}{\ds{enron}} %
\newcommand{\hos}{\ds{hospital-ward}} %
\newcommand{\hyp}{\ds{hypertext}} %
\newcommand{\infec}{\ds{infectious}} %
\newcommand{\pri}{\ds{primaryschool}} %

In this section we present our experimental results. We give the \dsd{} of
several real-world temporal graphs for several values for $\Delta$. Then we show the behavior of our
implementation of \bkdp{} (Algorithm~\ref{alg:bronkerdeltaPivot}) applied to these real-world temporal graphs and compare it to the algorithm implemented by \citeVLM.

\subsection{Setup and Statistics} 
We now give details of the implementation and the used reference algorithm, and introduce the data
sets we used in the experiments. Furthermore, we explain how the values of $\Delta$ were chosen, give some statistics for the data set, and calculate the \dsd{} of the data sets for the chosen values of~$\Delta$.
 
\paragraph{Implementation.} We implemented\footnote{Code freely available at \url{http://fpt.akt.tu-berlin.de/temporalcliques/} (GNU General Public License).}
\bkdp{} with slight modifications that allow the algorithm to use multiple pivot elements (we refer to this version as \bkdp*). Furthermore, we implemented a simple algorithm to compute the
\dsd{}. Both implementations are in Python~2.7.12 and all experiments were carried out on an Intel Xeon
E5-1620 computer with four cores clocked at 3.6\,GHz and~64\,GB RAM. We
did not utilize the parallel-processing capabilities although it
should be easy to achieve almost linear speed-up with growing number
of cores due to the simple nature of \bkdp{}. The
operating system was Ubuntu 16.04.4 with Linux kernel version 4.4.0-57. We compared
\bkdp* with the algorithm by \citeVLM{}
which was also implemented in Python. %
We modified their source
code\footnote{Code freely available at 
  \url{https://github.com/TiphaineV/delta-cliques} .} by removing the text output in their implementation in order to avoid speed differences. We call their algorithm Algorithm~VLM below. 

\paragraph{Data Sets.} 
  We chose several freely available real-world
temporal graphs aiming for an overview over
the different kinds of contexts in which such graphs arise, that is, an overview over different modes 
of communication and different kinds of
entities and environments in which this communication takes place. However, a focus is on temporal graphs based on physical proximity of individuals, since previous work on \dclq s also focused on these~\cite{viard2015computing,Viard2015Dyno}. The contexts and sources of our test set of temporal graphs are as follows:

\begin{itemize}
\item internet-router communication: \aus\ \cite{leskovec2005graphs},
\item email communication: \ema\ \cite{EMT2011},
\item social-network communication: \soc\ %
  \cite{opsahl2009clustering}, and
\item 
  physical-proximity\footnote{Available at
    \url{http://www.sociopatterns.org/datasets} .} between
  \begin{itemize}
  \item high school students: \hsa, \hsb, \hsc\ \cite{gemmetto2014mitigation,stehle2011high,fournet2014contact},
  \item patients and health-care workers: \hos~\cite{vanhems2013estimating}, 
  \item attendees of the ACM Hypertext 2009 conference: \hyp~\cite{isella2011s}, %
  \item attendees of the Infectious SocioPatterns event: \infec~\cite{isella2011s},
    and %
  \item children and teachers in a primary school: \pri~\cite{stehle2011high}. %
  \end{itemize}
\end{itemize}
Table~\ref{tab:stats} contains the number of vertices, edges, temporal resolution, and lifetime of the corresponding temporal graphs. As a time step we fixed one second for each of the data sets. \citeVLM{}, as the first work on enumerating \dC s, used the data set \hsb\ in their experiments.

\paragraph{Chosen values of $\Delta$.}
In order to
limit the influence of time scales in the data and
to make running times comparable between instances, as well as to
be able to present the results in a unified way, we chose the
$\Delta$-values as follows. We decided on a reference point of the
\emph{edge appearance rate} that is, of the average number of edges
per time step and we fixed a set of $\Delta$-values for this reference
point. For each considered instance we then scaled the reference
$\Delta$-values proportionally to the quotient of the reference edge
appearance rate and the edge appearance rate in the instance.

As the reference point we chose the edge appearance rate of $1/5$
edges per time
step; %
this value was chosen for convenience within the interval of
edge-appearance rates in the studied data sets (see
Table~\ref{tab:stats}). Since, intuitively, the $\Delta$-values of
interest in practice increase exponentially, we chose as reference
$\Delta$-values the numbers $0$ and $5^i$ for $i = 1, 2, \ldots$. As
mentioned, for each instance, these values are then multiplied by the
quotient of edge appearance rates. That is, if the instance has $m$
edges and lifetime~$\ell$, then we scaled the reference
$\Delta$-values by the factor $(1/5)/(m/\ell) = \ell/(5m)$. For
example, for \hsb{} we obtain the $\Delta$-values
$\{0, 80, 404, 2024, 10121, 50606,$ $253034, \ldots\}$. For reference,
recall that each time step in \hsb{} corresponds to one second (a day
has 86,000 seconds and a week has 604,800 seconds). In figures, we
refer to each scaled value of $\Delta$ by $\Delta \sim 5^i$ for some
concrete~$i$. \citet{Viard2015Dyno} used $\Delta$-values according to
$60$ seconds, $15$ minutes, $1$ hour and $3$ hours.

\begin{table}[t]
  \centering
  \caption{Statistics for the data sets used in our experiments.}
  \begin {tabular}{lrrrr}%
\toprule Instance&Vertices&Edges&Resolution&Lifetime (s)\\\midrule %
as733&\pgfutilensuremath {7{,}716}&\pgfutilensuremath {11{,}410{,}810}&1d&\pgfutilensuremath {67{,}824{,}000}\\%
facebook-like&\pgfutilensuremath {1{,}899}&\pgfutilensuremath {59{,}835}&1s&\pgfutilensuremath {16{,}736{,}181}\\%
highschool-2011&\pgfutilensuremath {126}&\pgfutilensuremath {28{,}560}&20s&\pgfutilensuremath {272{,}330}\\%
highschool-2012&\pgfutilensuremath {180}&\pgfutilensuremath {45{,}047}&20s&\pgfutilensuremath {729{,}500}\\%
highschool-2013&\pgfutilensuremath {327}&\pgfutilensuremath {188{,}508}&20s&\pgfutilensuremath {363{,}560}\\%
hospital-ward&\pgfutilensuremath {75}&\pgfutilensuremath {32{,}424}&20s&\pgfutilensuremath {347{,}500}\\%
hypertext&\pgfutilensuremath {113}&\pgfutilensuremath {20{,}818}&20s&\pgfutilensuremath {212{,}340}\\%
infectious&\pgfutilensuremath {10{,}972}&\pgfutilensuremath {415{,}912}&20s&\pgfutilensuremath {6{,}946{,}340}\\%
karlsruhe&\pgfutilensuremath {1{,}870}&\pgfutilensuremath {461{,}661}&1s&\pgfutilensuremath {123{,}837{,}267}\\%
primaryschool&\pgfutilensuremath {242}&\pgfutilensuremath {125{,}773}&20s&\pgfutilensuremath {116{,}900}\\\bottomrule %
\end {tabular}%

  \label{tab:stats}
\end{table}

\begin{table}[t]
  \caption{Static degeneracy and \dsd. Empty cells indicate that the lifetime of the temporal graph is smaller than the scaled $\Delta$-value.}
  \centering
  \begin {tabular}{lcccccc}%
\toprule Instance&Static&$\Delta =0$&$\sim 5^3$&$\sim 5^5$&$\sim 5^7$&$\sim 5^9$\\\midrule %
as733&\pgfutilensuremath {24}&\pgfutilensuremath {13}&\pgfutilensuremath {13}&\pgfutilensuremath {14}&\pgfutilensuremath {15}&\pgfutilensuremath {24}\\%
facebook-like&\pgfutilensuremath {20}&\pgfutilensuremath {1}&\pgfutilensuremath {3}&\pgfutilensuremath {6}&\pgfutilensuremath {19}&\\%
highschool-2011&\pgfutilensuremath {21}&\pgfutilensuremath {4}&\pgfutilensuremath {7}&\pgfutilensuremath {11}&\pgfutilensuremath {19}&\\%
highschool-2012&\pgfutilensuremath {18}&\pgfutilensuremath {4}&\pgfutilensuremath {5}&\pgfutilensuremath {6}&\pgfutilensuremath {12}&\\%
highschool-2013&\pgfutilensuremath {24}&\pgfutilensuremath {4}&\pgfutilensuremath {5}&\pgfutilensuremath {9}&\pgfutilensuremath {14}&\\%
hospital-ward&\pgfutilensuremath {22}&\pgfutilensuremath {4}&\pgfutilensuremath {6}&\pgfutilensuremath {11}&\pgfutilensuremath {18}&\\%
hypertext&\pgfutilensuremath {28}&\pgfutilensuremath {6}&\pgfutilensuremath {7}&\pgfutilensuremath {8}&\pgfutilensuremath {22}&\\%
infectious&\pgfutilensuremath {18}&\pgfutilensuremath {4}&\pgfutilensuremath {9}&\pgfutilensuremath {18}&\pgfutilensuremath {18}&\pgfutilensuremath {18}\\%
karlsruhe&\pgfutilensuremath {33}&\pgfutilensuremath {2}&\pgfutilensuremath {6}&\pgfutilensuremath {9}&\pgfutilensuremath {17}&\pgfutilensuremath {32}\\%
primaryschool&\pgfutilensuremath {47}&\pgfutilensuremath {4}&\pgfutilensuremath {4}&\pgfutilensuremath {10}&\pgfutilensuremath {31}&\\\bottomrule %
\end {tabular}%

  \label{tab:dsd}
\end{table}

\paragraph{\Dsd.}   The \dsdm{} for our set of instances are
shown in Table~\ref{tab:dsd} together with the static degeneracy of
the underlying static graph which has an edge whenever there is an
edge at some time step in the temporal graph. Clearly, as the value of
$\Delta$ increases, the \dsd{} approaches---and is upper-bounded
by---the static degeneracy. The static degeneracy of our
instances is small in comparison with the size of the graph. This
falls in line with the analysis by \citet{ELS13} for many real-world
graphs. Moreover, for many practically relevant values of $\Delta$ the
\dsd{} is still significantly smaller. For example, in the instance \hsb, the
scaled value of $\Delta$ corresponding to $5^3$ equals~2204 time steps
(seconds) and the corresponding \dsd{} is~5. This indicates that \dsd\ can be a very promising (that is, also small) parameter when designing and analyzing algorithms for temporal graphs. 

We computed the \dsdm{} using a straightforward approach. We
iteratively computed %
for each $\Delta$-long time interval the graph induced by the edges in
that time interval. For each of these graphs we computed the static
degeneracy using an implementation from the NetworkX python
library~\cite{HSS08}. This approach is rather inefficient. For
example, it took about seven hours to compute the \dsd\ for \ema\ with $\Delta \sim 5^5$
(equating to a $\Delta$ value of about two hours).

\subsection{Results and Running Times}
  We now study the efficiency of \bkdp*,
evaluate pivoting strategies, and compare the result to
Algorithm~VLM.

\paragraph{Pivoting.}   Generally we observed that pivoting plays a
negligible role when $\Delta$ is small compared to the overall
lifetime of the graph%
, that is, when $\Delta$ is less than roughly one
third of the lifetime. In this case, pivoting has almost no effect on
the running time and the number of recursive calls. For larger values of $\Delta$, however, pivoting can make a clear difference depending on the type of temporal graph.

We tested five strategies for selecting a set of pivots from $P$ in
\bkdp*. Call a set of pivots is \emph{maximal} if the interval of each
element from~$P$ overlaps with at least one pivot. We tested the following variants of pivot sets: 
\begin{enumerate}
\item[1A)] a single arbitrary pivot,
\item[1G)] a single pivot maximizing the number of elements removed from~$P$,
\item[MA)] an arbitrary maximal set of pivots (pivots picked one-by-one arbitrarily),
\item[MG)] a maximal set of pivots (pivots picked one by one according to the
  maximum number of further elements removed from~$P$), and
\item[MM)] a set of pivots which maximizes the number of
  elements removed from~$P$. 
\end{enumerate}
Clearly, each strategy has its own
trade-off between the time needed to compute the pivots and the
possible reduction in recursive calls.

\begin{figure}[t]
  \centering
  \includegraphics{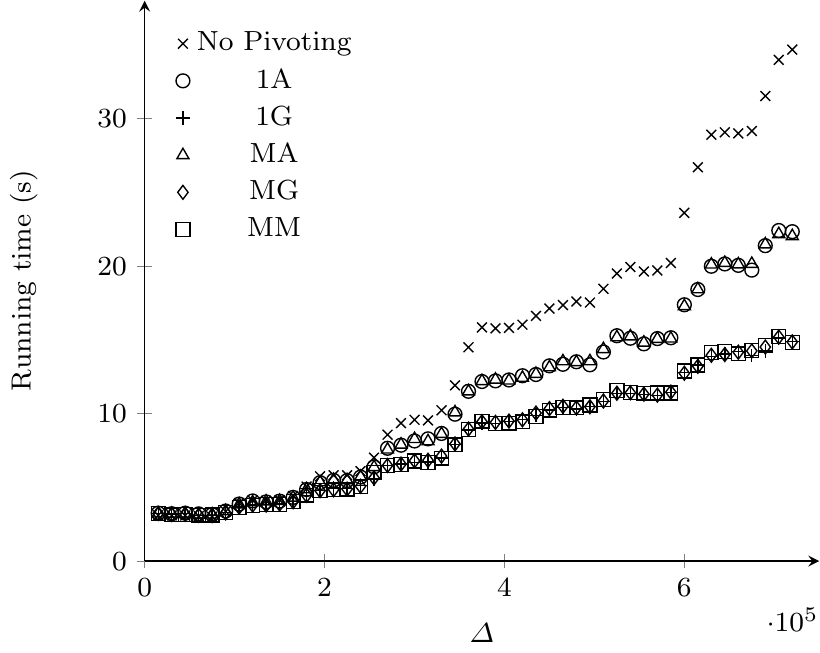}
  \caption{Running time for different pivoting strategies on \hsb{}.}
  \label{fig:pivot-run-hsb}
\end{figure}

\begin{figure}[t]
  \centering
  \includegraphics{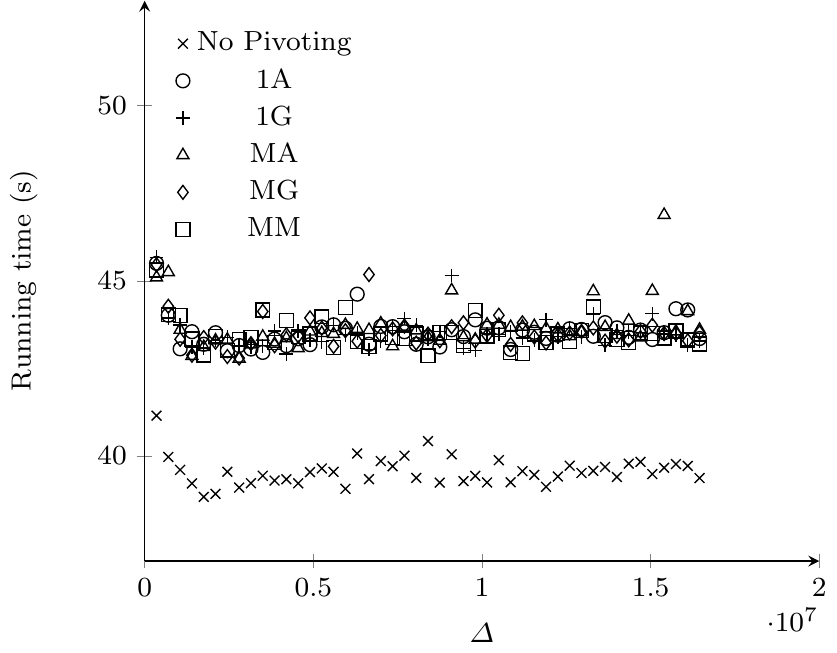}
  \caption{Running time for different pivoting strategies on \soc{}.}
  \label{fig:pivot-run-soc}
\end{figure}

Running times are given for \hsb{} in Figure~\ref{fig:pivot-run-hsb}
with $\Delta$ between~15,000 and 725,000. We note that running times
for some very small values of $\Delta$ below 15,000 are larger than
30\,s and hence do not fit in the chart. We consider this phenomenon
more closely below. For $\Delta \leq {}$15,000 there is no appreciable
difference between the pivoting strategies. In terms of relative
difference between pivoting strategies, \hsb{} seems to be a
representative example. Strategies~1G and~MG seem to be the best
options: they do not incur much overhead compared to no pivoting for
small~$\Delta$ and yield strong running time improvements for larger
$\Delta$. In comparison to no pivoting, strategies~1G and~MG achieve a
60\,\% reduction in recursive calls for $\Delta$-values of around
$7 \cdot 10^6$ in \hsb{}. Since the running times of strategy~1G
and~MG are so close to each other we conclude that in most cases there
is only one important pivot that should be selected. We were
surprised to see that maximizing the overall number of elements
removed from~$P$ via the pivot set (strategy~MM) results in slightly
worse running times and slightly larger numbers of recursive
calls%
. The number of elements that are removed by a
pivot in one recursive call of the algorithm ranges between one and 14
while many of the calls remove two to four elements%
. Notice that occasional reduction by ten or more
elements can substantially decrease the search space, because in general its size
depends exponentially on the size of~$P$.

Figure~\ref{fig:pivot-run-soc} shows running times for \soc{}. On this
graph, pivoting seldom removes more than one element from the
candidate set~$P$ in one call of the recursive procedure. Hence, for this instance, pivoting
mainly incurs overhead for computing the pivots, but do not
substantially decrease the search space. We consequently observe about
10\,\% slower running times, regardless of the pivoting strategy. 

In conclusion, strategy~1G offers the best trade-off between additional
running time spent with computing the pivot(s) and running time saved
due to decreased number of recursive calls. Overall, the the possible
benefits seem to outweigh the overhead incurred by pivoting on some
instances. All remaining experiments were thus carried out with
strategy~1G.

\paragraph{Running Times and Comparison with Algorithm VLM.}
\begin{table*}[!p]%
  \caption{$\Delta$-clique statistics and running times: $|\cal C|$ denotes the number of maximal $\Delta$-cliques, $s$ denotes the maximum $\Delta$-clique size, $\ell$ the maximum $\Delta$-clique lifetime divided by~$10^5$, $t_\text{BKD}$ and $t_\text{VLM}$ denote the running time in seconds of \bkd* and Algorithm VLM, respectively. Empty cells represent an exceeded running time limit of one hour.}%
  \centering%
  \begin {tabular}{lrrrrr}%
\toprule & \multicolumn {5}{c}{$\Delta = 0$} \\ \cmidrule (r){2-6}Instance&$|\cal C|$&$s$&$\ell $&$t_\text {BKD}$&$t_\text {VLM}$\\\midrule %
facebook-like&\pgfutilensuremath {61{,}648}&\pgfutilensuremath {2}&\pgfutilensuremath {1{,}674}&\pgfutilensuremath {169}&\pgfutilensuremath {12}\\%
highschool-2011&\pgfutilensuremath {26{,}510}&\pgfutilensuremath {5}&\pgfutilensuremath {27}&\pgfutilensuremath {131}&\pgfutilensuremath {7}\\%
highschool-2012&\pgfutilensuremath {42{,}285}&\pgfutilensuremath {5}&\pgfutilensuremath {73}&\pgfutilensuremath {248}&\pgfutilensuremath {12}\\%
highschool-2013&\pgfutilensuremath {172{,}362}&\pgfutilensuremath {5}&\pgfutilensuremath {36}&\pgfutilensuremath {1{,}952}&\pgfutilensuremath {118}\\%
hospital-ward&\pgfutilensuremath {27{,}910}&\pgfutilensuremath {5}&\pgfutilensuremath {35}&\pgfutilensuremath {370}&\pgfutilensuremath {14}\\%
hypertext&\pgfutilensuremath {19{,}150}&\pgfutilensuremath {6}&\pgfutilensuremath {21}&\pgfutilensuremath {85}&\pgfutilensuremath {5}\\%
infectious&\pgfutilensuremath {349{,}787}&\pgfutilensuremath {5}&\pgfutilensuremath {695}&\pgfutilensuremath {1{,}530}&\pgfutilensuremath {2{,}515}\\%
karlsruhe&&&&&\pgfutilensuremath {1{,}494}\\%
primaryschool&\pgfutilensuremath {107{,}121}&\pgfutilensuremath {5}&\pgfutilensuremath {12}&\pgfutilensuremath {995}&\pgfutilensuremath {147}\\\bottomrule %
\end {tabular}%
\\
  \vspace{.3cm}
  \begin {tabular}{lrrrrr}%
\toprule & \multicolumn {5}{c}{$\Delta \sim 5^3$}\\ \cmidrule (r){2-6}Instance&$|\cal C|$&$s$&$\ell $&$t_\text {BKD}$&$t_\text {VLM}$\\\midrule %
facebook-like&\pgfutilensuremath {33{,}876}&\pgfutilensuremath {4}&\pgfutilensuremath {1{,}675}&\pgfutilensuremath {70}&\pgfutilensuremath {1{,}141}\\%
highschool-2011&\pgfutilensuremath {7{,}394}&\pgfutilensuremath {7}&\pgfutilensuremath {27}&\pgfutilensuremath {5}&\pgfutilensuremath {153}\\%
highschool-2012&\pgfutilensuremath {9{,}501}&\pgfutilensuremath {6}&\pgfutilensuremath {73}&\pgfutilensuremath {8}&\pgfutilensuremath {236}\\%
highschool-2013&\pgfutilensuremath {57{,}121}&\pgfutilensuremath {6}&\pgfutilensuremath {36}&\pgfutilensuremath {178}&\pgfutilensuremath {1{,}990}\\%
hospital-ward&\pgfutilensuremath {8{,}694}&\pgfutilensuremath {7}&\pgfutilensuremath {35}&\pgfutilensuremath {15}&\pgfutilensuremath {226}\\%
hypertext&\pgfutilensuremath {6{,}345}&\pgfutilensuremath {7}&\pgfutilensuremath {21}&\pgfutilensuremath {8}&\pgfutilensuremath {107}\\%
infectious&\pgfutilensuremath {134{,}787}&\pgfutilensuremath {9}&\pgfutilensuremath {695}&\pgfutilensuremath {1{,}195}&\\%
karlsruhe&&&&&\\%
primaryschool&\pgfutilensuremath {83{,}314}&\pgfutilensuremath {9}&\pgfutilensuremath {12}&\pgfutilensuremath {83}&\\\bottomrule %
\end {tabular}%
\\
  \vspace{.3cm}
  \begin {tabular}{lrrrrr}%
\toprule & \multicolumn {5}{c}{$\Delta \sim 5^5$}\\ \cmidrule (r){2-6}Instance&$|\cal C|$&$s$&$\ell $&$t_\text {BKD}$&$t_\text {VLM}$\\\midrule %
facebook-like&\pgfutilensuremath {23{,}247}&\pgfutilensuremath {5}&\pgfutilensuremath {1{,}709}&\pgfutilensuremath {47}&\\%
highschool-2011&\pgfutilensuremath {7{,}760}&\pgfutilensuremath {10}&\pgfutilensuremath {28}&\pgfutilensuremath {4}&\\%
highschool-2012&\pgfutilensuremath {7{,}536}&\pgfutilensuremath {7}&\pgfutilensuremath {75}&\pgfutilensuremath {4}&\pgfutilensuremath {1{,}365}\\%
highschool-2013&\pgfutilensuremath {29{,}752}&\pgfutilensuremath {8}&\pgfutilensuremath {37}&\pgfutilensuremath {23}&\\%
hospital-ward&\pgfutilensuremath {10{,}869}&\pgfutilensuremath {12}&\pgfutilensuremath {36}&\pgfutilensuremath {7}&\\%
hypertext&\pgfutilensuremath {7{,}459}&\pgfutilensuremath {7}&\pgfutilensuremath {23}&\pgfutilensuremath {4}&\pgfutilensuremath {2{,}193}\\%
infectious&\pgfutilensuremath {163{,}162}&\pgfutilensuremath {16}&\pgfutilensuremath {697}&\pgfutilensuremath {1{,}277}&\\%
karlsruhe&\pgfutilensuremath {235{,}684}&\pgfutilensuremath {9}&\pgfutilensuremath {12{,}417}&\pgfutilensuremath {1{,}235}&\\%
primaryschool&\pgfutilensuremath {508{,}430}&\pgfutilensuremath {20}&\pgfutilensuremath {15}&\pgfutilensuremath {890}&\\\bottomrule %
\end {tabular}%

  \label{tab:holist}
\end{table*}
We experimented with \bkdp* (Algorithm~\ref{alg:bronkerdeltaPivot}) using pivoting strategy~1G and with 
Algorithm~VLM for $\Delta = 0$ and $\Delta \sim 5^i$ with $i = 1, 3, 5, 7, 9$
(where the lifetime allowed such values of $\Delta$). An excerpt of
the results is given in Table~\ref{tab:holist}. Clearly, larger
instances with more vertices or edges demand a longer running
time. However, even large instances like \infec{} can still be solved
within one hour.

From our theoretical results in
Section~\ref{section:bronKerboschDelta} we expected that the running
time of \bkdp* increases exponentially with
growing \dsd{}. As the \dsd{} grows very slowly with
increasing~$\Delta$ (see Table~\ref{tab:dsd}), we expected a
corresponding moderate growth in running time with respect
to~$\Delta$. For larger $\Delta$, this is consistent with the
experimental results, as shown in Figures~\ref{fig:pivot-run-hsb},~\ref{fig:pivot-run-soc} and
Table~\ref{tab:holist}. However, for (very) small $\Delta$ we often observe an
initial spike in the running time (and number of $\Delta$-cliques)
which then subsides. This is also shown in Figure~\ref{fig:smalld-hsb}.
A possible explanation for this spike is that, for small~$\Delta$, the $\Delta$-neighborhood of many vertices becomes very
fragmented, leading to large candidate sets~$P$ in the algorithm (although
the size of $P$ is still linear in the input size for constant
\dsd). Furthermore, if $\Delta$ is small, then many singleton edges may
form maximal $\Delta$-cliques themselves. These $\Delta$-cliques then
get taken up into larger maximal $\Delta$-cliques when $\Delta$
increases, which decreases the number of $\Delta$-cliques and running
times for \bkdp*.

On \soc\ our algorithm notably is comparably efficient given the
relatively large size (see
Figures~\ref{fig:pivot-run-soc} and~\ref{fig:smalld-soc}). Furthermore, the number of \dclq s
does not seem to vary strongly with changing values of~$\Delta$. These
two facts may hint at some special structure that is present in temporal graphs based on online
social networks, in addition to small \dsd.

  Algorithm~VLM is usually faster than \bkdp* for small
values of~$\Delta$ below the $\Delta \sim 5^3$ threshold. Starting from there,
however, \bkdp* outperforms
Algorithm VLM with running times smaller by at least one order of
magnitude and up to three orders of magnitude (see
Table~\ref{tab:holist}). In terms of main memory, 385\,MB is the
maximum used by \bkdp* over all solved
instances, attained on \infec{} for $\Delta = 0$. On this instance,
Algorithm VLM uses~494\,MB and often more than 1\,GB.

\begin{figure}[t]
  \centering
  \includegraphics{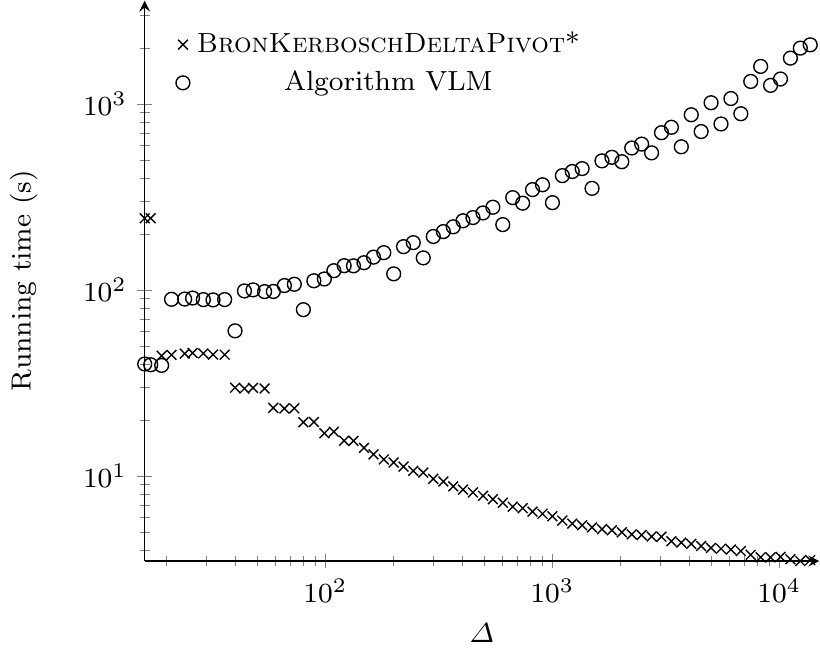}
  \caption{Running time vs.\ $\Delta$ on \hsb{}.}
  \label{fig:smalld-hsb}
\end{figure}

\begin{figure}[t]
  \centering
  \includegraphics{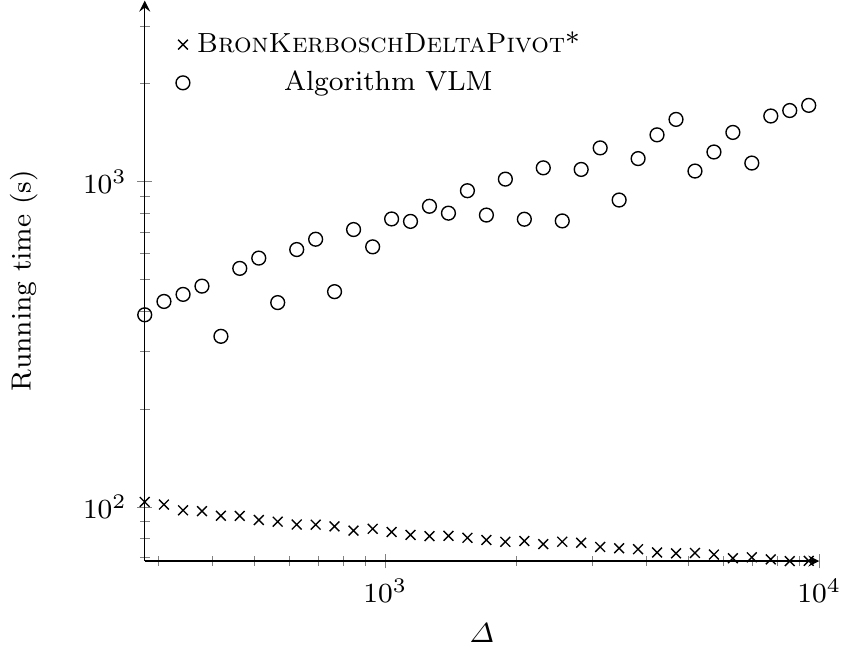}
  \caption{Running time vs.\ $\Delta$ on \soc{}.}
  \label{fig:smalld-soc}
\end{figure}

Finally we mention that, when increasing the time limit to six hours,
\bkdp* can solve all instances of \ema{}
for $\Delta = 0$ and $\Delta \sim 5^i$ for $i = 1, 3, 5, 7, 9$ wherein the last value of~$\Delta$
involves enumerating $43 \cdot 10^6$ maximal $\Delta$-cliques.

\section{Conclusion and Outlook}
\label{sec:conclusion}
We studied the algorithmic complexity of enumerating $\Delta$-cliques in temporal graphs. We adapted the Bron-Kerbosch algorithm~(\cite{bron1973algorithm}), including the procedure of pivoting to reduce the number of recursion calls, to the temporal setting and provided a theoretical analysis. For the theoretical analysis, we formalized and employed the concept of $\Delta$-slice degeneracy which may be a useful parameter when analyzing problems in sparse temporal graphs.

In experiments on real-world data sets, we showed that our algorithm is notably faster than the first approach for enumerating all maximal $\Delta$-cliques in temporal graphs due to \citet{Viard2015Dyno,viard2015computing}. Our experimental results further reveal that pivoting can notably decrease the running time for large values of~$\Delta$. Furthermore, we measured the $\Delta$-slice degeneracy for different $\Delta$-values and showed that it is reasonably small in many real-world data sets. 

As to future research, an algorithmic challenge is to find a more efficient way to compute the \dsd\ of a given temporal graph, perhaps via different characterizations as in the case of static graphs. See \cite{ELS13} for an account of several equivalent definitions of the degeneracy of a static graph.
Regarding the adapted version of the Bron-Kerbosch algorithm, our theoretical analysis (based on the $\Delta$-slice degeneracy parameter) of the running time still leaves room for improvement. In particular, we leave the impact of pivoting on the running time upper bound as an open question for future research. It furthermore makes sense to try and implement further improved branching rules on top of pivoting. This was also successful for the static Bron-Kerbosch algorithm \cite{Nau16}. Another interesting question is whether an analogue to the degeneracy ordering can be defined in the temporal setting and, if so, whether it can be used to further improve the algorithm.

\paragraph*{Acknowledgements.}
  Anne-Sophie Himmel, Hendrik Molter and Manuel Sorge were partially supported by DFG, project DAPA (NI~369/12). Manuel Sorge  gratefully acknowledges support by the People Programme (Marie Curie Actions) of the European Union's Seventh Framework Programme (FP7/2007-2013) under REA grant agreement number 631163.11 and by the Israel Science Foundation (grant no. 551145/14).
  
  We are grateful to two anonymous SNAM reviewers whose feedback helped to significantly improve the presentation and to eliminate some bugs and inconsistencies.

\bibliographystyle{abbrvnat}      %
\bibliography{literature}   %

\end{document}